\documentclass[sigconf, nonacm]{acmart}

\usepackage[ruled, vlined, linesnumbered]{algorithm2e}
\usepackage{subfig}
\usepackage{pifont}
\usepackage{xspace}

\newcommand{\ie}{{\emph{i.e.,}}\xspace}
\newcommand{\eg}{\emph{e.g.,}\xspace}

\newcommand{\etc}{\emph{etc.}\xspace}

\newcommand{\tableincell}[1]{\begin{tabular}[x]{@{}c@{}}#1\end{tabular}}

\newtheorem{lemma}{Lemma}
\newtheorem{theorem}{Theorem}

\newtheorem{example}{Example}
\newtheorem{definition}{Definition}

\newcommand{\stitle}[1]{\vspace{1.6ex}\noindent{\bf #1}}
\newcommand{\eat}[1]{{}}

\newcommand{\figref}[1]{{Fig.~\ref{#1}}}
\newcommand{\secref}[1]{{Sec.~\ref{#1}}}

\newcommand{\tabref}[1]{{Table~\ref{#1}}}

\newcommand{\lemref}[1]{{Lemma~\ref{#1}}}
\newcommand{\algref}[1]{{Algorithm~\ref{#1}}}

\newcommand{\egref}[1]{{Example~\ref{#1}}}
\newcommand{\corref}[1]{{Corollary~\ref{#1}}}
\newcommand{\defref}[1]{{Definition~\ref{#1}}}

\newcommand{\todo}[1]{\textbf{\color{red}{TODO: #1} }}

\newcommand{\ourindex}{{DLH}}

\newcommand{\REVIEWEAT}[1]{{}}
\newcommand\VLDBSUB

\newcommand\vldbavailabilityurl{https://github.com/TaoLbr1993/DLH}

\begin{document}
\title{Approximate Nearest Neighbor Search with Graph Range Filters}

\settopmatter{authorsperrow=4}
\author{Qian Tao}
\affiliation{%
  \institution{Beihang University}
  \city{Beijing}
  \state{China}
}
\email{qiantao@buaa.edu.cn}
\author{Yuntao Jiang}
\affiliation{%
  \institution{Beihang University}
  \city{Beijing}
  \state{China}
}
\email{yuntaojiang@buaa.edu.cn}
\author{Yongxin Tong}
\authornote{Yongxin Tong is the corresponding author.}
\affiliation{%
  \institution{Beihang University}
  \city{Beijing}
  \state{China}
}
\email{yxtong@buaa.edu.cn}

\author{Yu Sun}
\affiliation{%
  \institution{Nankai University}
  \city{Tianjin}
  \state{China}
}
\email{sunyu@nankai.edu.cn}

\begin{abstract}
Vector databases have become a fundamental component for high-dimensional vector retrieval in artificial intelligence applications. Recent research has focused on \emph{filtered approximate nearest neighbor search} (filtered ANN), which involves retrieving the nearest vectors that satisfy a given attribute-based filter. However, existing filters are generally limited to numerical range constraints or categorical existence checks, which restricts their applicability in more complex, real-world scenarios. 
In this paper, we investigate filtered ANN using \emph{graph range filters}, where the retrieved vectors must be within a specified distance from the query node in a predefined \emph{filter graph}. To address this problem, we propose \textsc{DLH}, a \underline{D}istance-aware \underline{L}abeling index with \underline{H}ashing compression. \textsc{DLH} creates distance-aware labeling sets to enable efficient graph range filters via the simplified set intersection operations.
Large labeling sets are further compressed into Bloom filters to improve query efficiency in {\ourindex}. 
Furthermore, recognizing that the query node is always involved in in-range queries of the graph range filters, we enhance {\ourindex} by memoizing the intermediate hashing index for the query node, yielding an optimized version called \textsc{DLH-M}. 
Experimental evaluations on diverse datasets demonstrate that \textsc{DLH} and \textsc{DLH-M} improve throughput by up to 70.3\%, and could maintain recall rates over $98.5\%$ with limited extra storage, validating the practical availability of the proposed solution.
\end{abstract}

\maketitle



\ifdefempty{\vldbavailabilityurl}{}{
\vspace{.3cm}
\begingroup\small\noindent\raggedright\textbf{Artifact Availability:}\\
The source code has been made available at \url{\vldbavailabilityurl}.
\endgroup
}

\section{Introduction}
\label{sec:intro}

As a fundamental component of modern artificial intelligence, \emph{vector databases} have attracted considerable attention due to their powerful vector search capabilities, and have become a core component of Retrieval Augmented Generation (RAG) systems~\cite{kdd24ragsur,arxiv23ragsur}.
They underpin a wide range of Intelligent tasks, including domain knowledge question answering~\cite{corr24kgrag}, text generation~\cite{sigir25ragtext}, image generation~\cite{corr25cmrag}, {\etc}
The core functionality of vector databases lies in approximate nearest neighbor (ANN) search, which retrieves the semantically nearest vectors to a given query vector.

\begin{figure}[t!]
\centering

\includegraphics[height=0.22\textwidth]{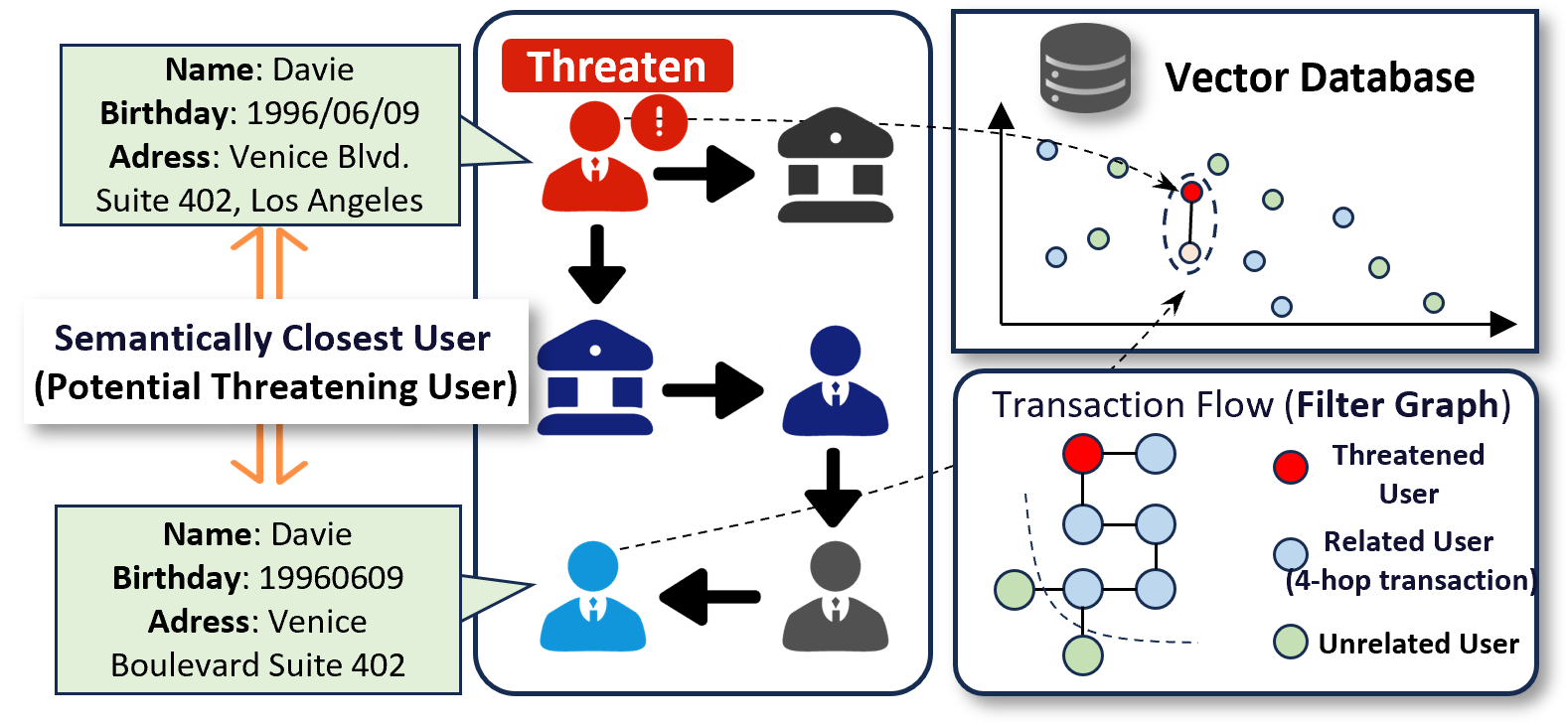}
\caption{ANN search with graph range filters.}
\label{fig:eg-aml}
\vspace{-8pt}
\end{figure}

In practical applications, vector data is frequently coupled with diverse attributes, necessitating \emph{filtered ANN search} to retrieve the nearest vectors that strictly adhere to specific attribute-based filters.
For instance, a product search on an e-commercial platform may require identifying similar items that were sold within the past three weeks.
While current filtering criteria effectively handle numerical ranges~\cite{vldb24unify,sigmod24irange,sigmod24serf} and categorical existence checks~\cite{icde25tag,www23diskann,sigmod25rwalks}, they are typically defined on oversimplified or low-cardinality attributes. This limitation constrains their applicability in complex, real-world filtered ANN search scenarios.

Real-world applications require filtered ANN search over complicated structured attributes. In scenarios such as graph-based RAG~\cite{corr24graphrag,sigmod25tigervec}, agent memory management~\cite{corr25mem0}, and anti-money laundering~\cite{access24aml}, similarity is evaluated among vectors that are constrained by a pre-defined \emph{filter graph}. This motivates a novel yet practical filtered ANN search setting. 
In this setting, the attributes associated with vectors are derived from \emph{graph-structured} data, rather than independent labels. To illustrate the necessity of this formulation, consider the following motivating example.

\begin{example}
\label{eg:aml}
As illustrated in \figref{fig:eg-aml}, consider an anti-money laundering (AML) system designed to identify suspicious users by analyzing capital flows. Each bank account, along with its metadata, is encoded as an embedding.
Notably, embeddings of accounts belonging to the same user or those involved in similar transaction patterns tend to exhibit higher proximity in vector space~\cite{cikm25dualgraph}. 
\eat{Given a target account identified as a potential threat, the bank seeks to retrieve related accounts that satisfy two simultaneous criteria: (1) topological reachability, i.e., they are within $r$ transaction hops from the target account in the transaction graph, and (2) semantic similarity, i.e., their embeddings are among the nearest neighbors to the target's vector.}
    Given a potentially threatening account represented by its embedding vector, the bank retrieves accounts that satisfy two criteria: (1) topological reachability, {\ie} they are within $r$ transaction hops from the target accounts in the transaction graph, and (2) semantic similarity, {\ie} their embeddings are among the nearest neighbors to the target's vector. 
\end{example}

\egref{eg:aml} motivates the Approximate Nearest Neighbor search with Graph Range filters (ANNGR) problem, where each vector is associated with a node in a filter graph. The retrieved vectors are required to be not only similar to the query vector but also located within a $r$-hop neighborhood of the query node. 

Different from existing filtered ANN search settings, ANNGR introduces two key challenges, as illustrated in \figref{fig:anngr-hard}.


 \noindent \underline{\textbf{C1:} \emph{Explosive Volume of Candidate Attributes.}} In the ANNGR problem, the attributes, specifically the corresponding nodes in the filter graph, form a candidate space with a cardinality comparable to the total number of vectors.
 Such a massive attribute space renders the construction of attribute-independent ANN indices (as exploded in \cite{sigmod24serf,icde25time}) computationally prohibitive, as it would require maintaining an impractical number of separate indexing structures.
 
\noindent\underline{\textbf{C2:}  \emph{Irregular and Unpredictable Candidate Distribution.}}
Due to the inherent complexity of graph topologies, the $r$-hop neighborhoods of different query nodes, along with their potential overlaps, are highly irregular and stochastic.
This unpredictability precludes the feasibility of precomputing or exploiting these overlaps to design specialized indices, as is often done in simpler categorical filtering.


To overcome the above limitations, this paper proposes {\ourindex}, a \underline{D}istance-aware \underline{L}abeling index with \underline{H}ashing compression, specifically engineered for the ANNGR problem.
{\ourindex} integrates conventional ANN indices ({\eg} HNSW~\cite{pami20hnsw}) with a novel, lightweight distance-aware labeling index.
Within this framework, graph range filters are efficiently executed via a limited number of set intersection operations.
To optimize memory footprint and computational throughput, 
{\ourindex} compresses extensive labeling sets into fixed-size bit vectors using Bloom filters, thereby enabling approximate set intersections with bounded error rates.
Furthermore, observing that the query node consistently serves as the source for numerous in-range queries, we introduce {\ourindex}-M.
This variant memoizes the pre-computed hashing indices of query nodes prior to the retrieval phase, thereby substantially boosting overall query efficiency by eliminating redundant computations.

\begin{figure}[t!]
\centering

\includegraphics[height=0.25\textwidth]{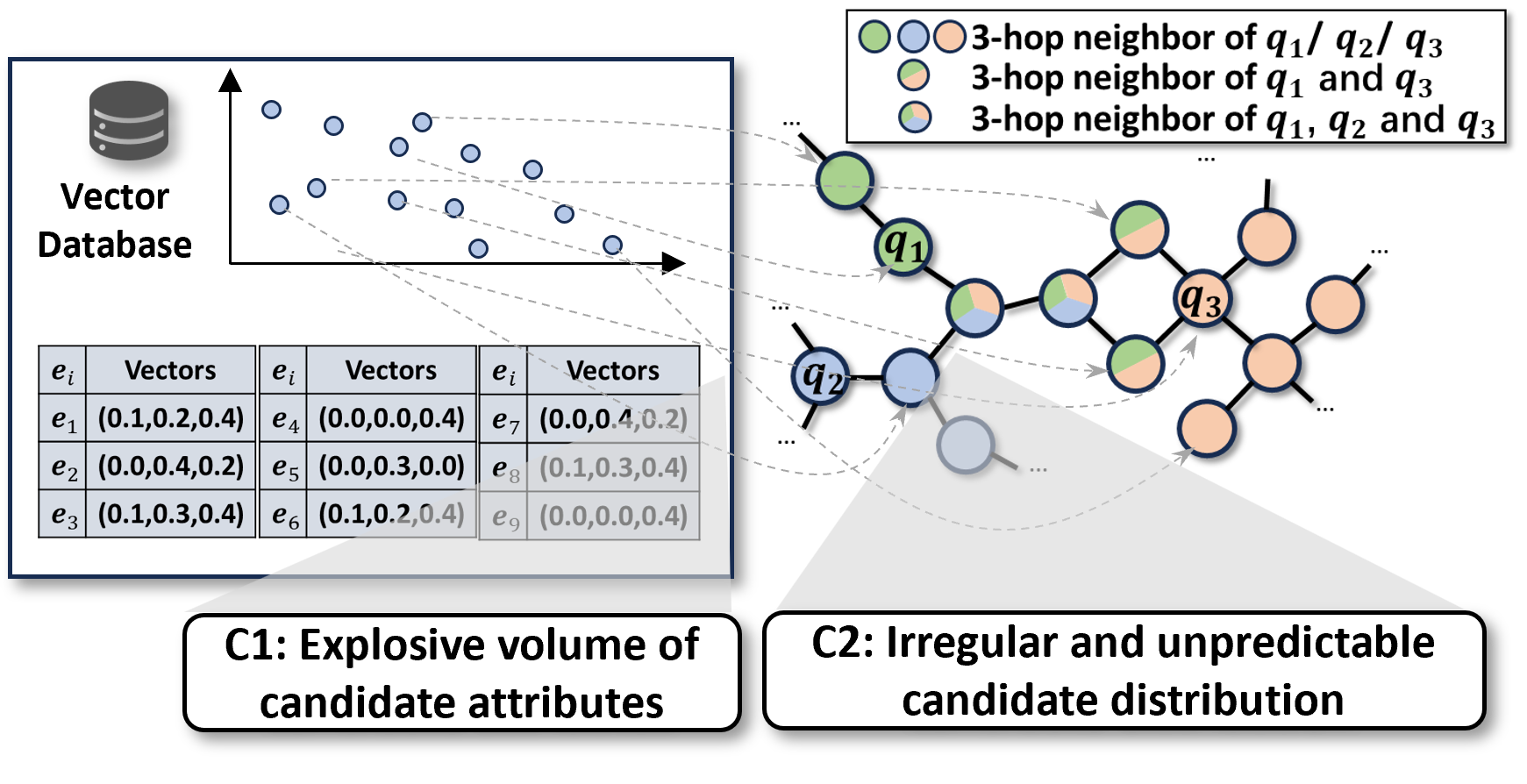}
\caption{Challenges of ANNGR problem.}
\label{fig:anngr-hard}
\vspace{-4pt}
\end{figure}

In summary, this paper makes the following key
contributions.
\begin{itemize}
    \item To the best of our knowledge, we are the first to formally define and investigate the Approximate Nearest Neighbor search with Graph Range filters (ANNGR) problem.
    \item We propose {\ourindex}, a distance-aware labeling index integrated with hashing compression. This framework enables efficient ANNGR processing in the post-filtering manner.
    \item Leveraging the asymmetric nature of in-range predicates in graph-structured data, we introduce {\ourindex}-M. This variant strategically optimizes the post-filtering pipeline by memoizing query-side indices, leading to superior query performance.
    \item Extensive experiments on diverse datasets demonstrate that our proposed algorithms outperform baseline methods, improving ANNGR query efficiency by up to 70.3\%.
\end{itemize}

The remainder of this paper is organized as follows. In \secref{sec:pre}, we present the preliminaries. In \secref{sec:alg}, we introduce the proposed {\ourindex}, followed by {\ourindex}-M in \secref{sec:par}. The experimental evaluation is presented in \secref{sec:exp}. We review related work in \secref{sec:rel} and conclude the paper in \secref{sec:con}.
\section{Preliminaries}
\label{sec:pre}

\subsection{Problem Definition}
\label{subsec:def}
Define a set of vectors $D=\{\mathbf{e}_1,...,\mathbf{e}_n\}$, where each $\mathbf{e}_i$ is a $dim$-dimensional vector and $n$ denotes the cardinality of $D$.
For any two vectors $\mathbf{e}_i$ and $\mathbf{e}_j$, we denote their distance by $d(\mathbf{e}_i,\mathbf{e}_j)$.
Common distance functions include Euclidean distance and cosine similarity-based distance~\cite{icde25time}.

Each vector $\mathbf{e}_i\in D$ is associated with a node $\mathbf{e}_i.v\in V$ in a \emph{filter graph} $G=(V,E)$, where $V=\{v^*_1,\dots,v^*_{n_f}\}$ and $E\subseteq V\times V$ denote the sets of nodes and edges, respectively.
When the context is clear, we use $v_i$ interchangeably to refer to $\mathbf{e}_i.v$.
For any node $u\in V$, $G[u]$ denotes its set of neighbors in $G$, and $d_G(u,v)$ represents the shortest-path distance between nodes $u$ and $v$ in the filter graph.
We assume that $G$ is undirected and that multiple vectors in $D$ may be associated with the same node in $G$.

\begin{definition}[Nearest Neighbor Search with Graph Range Filters (NNGR)]
\label{def:nngr}
Given a set of vectors $D$, where each vector $\mathbf{e}_i\in D$ is associated with a node $v_i$ in the filter graph $G$, a query $(\mathbf{e}_q,v_q, k, r)$, and an integer $k$, the Nearest Neighbor search with Graph Range filters (NNGR) problem returns a set $R\subseteq D$ of $k$ vectors satisfying the following two conditions: 
\begin{itemize}
    \item $R$ consists of the $k$ nearest neighbors to the query vector $\mathbf{e}_q$ among all vectors in $D$.
    Formally, $|R|=k$ and for every $ \mathbf{e}_i\in R$ and $\mathbf{e}_j \in D\setminus R$, it holds that $d(\mathbf{e}_q,\mathbf{e}_i)\leq d(\mathbf{e}_q,\mathbf{e}_j)$;
    \item the node $v_i$ associated with each $\mathbf{e}_i\in R$ satisfies $d_G(v_q,v_i)\leq r$, {\ie} it lies within $r$ hops of $v_q$ in $G$.
\end{itemize}
\end{definition}

\begin{example}
\figref{fig:eg-def} illustrates an example for the ANNGR problem involving $9$ vectors.
The left part of \figref{fig:eg-def} lists the vectors $\mathbf{e}_1,\dots, \mathbf{e}_9$ along with the query vector $\mathbf{e}_q$, ordered by increasing vector distance to $\mathbf{e}_q$.
The right part shows the corresponding filter graph $G$, where each of the $10$ vectors (including $\mathbf{e}_q$) is associated with a node in $G$.
Consider the query $(\mathbf{e}_q, v_q, 1, 2)$ which seeks the single nearest neighbor to $\mathbf{e}_q$ among vectors whose associated nodes lie within in a shortest-path distance of at most $2$ hops from $v_q$ in the filter graph.
Although $\mathbf{e}_4$ is the closest vector to $\mathbf{e}_q$ in embedding space, its associated node lies more than $2$ hops away from $v_q$, violating the graph-range constraint.
Consequently, $\mathbf{e}_2$ is the closest valid candidate within the allowed range and should be returned as the result.
\end{example}

Following recent efforts in approximate nearest neighbor search queries~\cite{icde25time,sigmod24serf,sigmod24acorn,vldb25navix}, this paper focuses on the \emph{approximate} nearest neighbor search with graph range filters.

\begin{definition}[Approximate Nearest Neighbor Search with Graph Range Filters (ANNGR)]
    Given a vector set $D$, a filter graph $G$, and an ANNGR query $(\mathbf{e}_q,v_q,k,r)$, the Approximate Nearest Neighbor search with Graph Range filters (ANNGR) problem returns a set $R\subseteq D$ of $k$ vectors such that:
    \begin{itemize}
    \item the vectors in $R$ approximately minimize the distances to the query vector $\mathbf{e}_q$, i.e., they form an approximate top-$k$ set among all vectors in $D$ whose associated nodes lie within $r$ hops of $v_q$;
    \item the node $v_i \in V$ associated to each vector $\mathbf{e}_i \in R$  satisfies $d_G(v_q, v_i) \leq r$.
    
    \end{itemize}
\end{definition}

\begin{figure}[t!]
\centering

\includegraphics[height=0.28\textwidth]{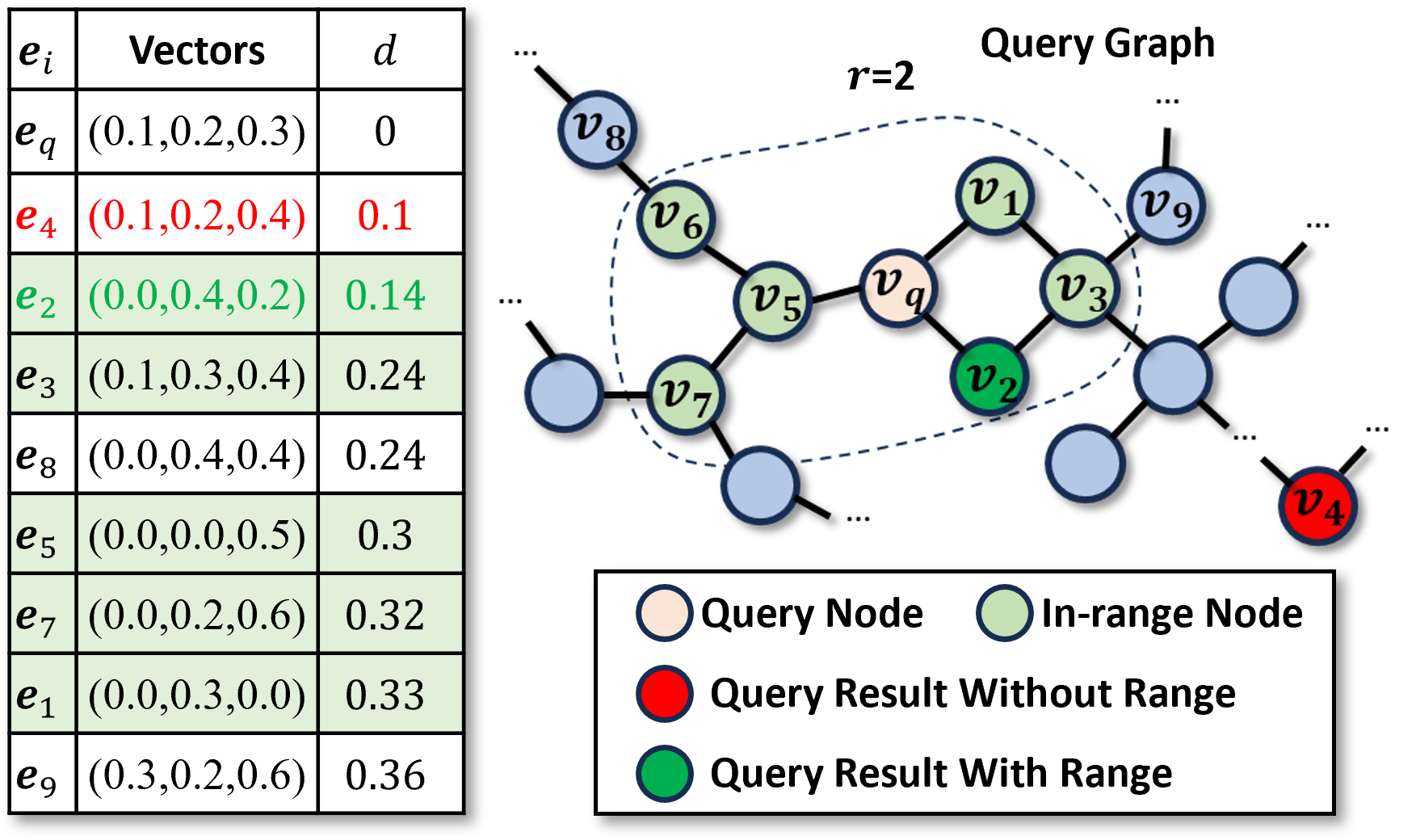}
\caption{An example of the ANNGR problem.}
\label{fig:eg-def}
\end{figure}

Following prior works, we use recall, defined as $\frac{|R\cap R^*|}{k}$, to evaluate the quality of an ANNGR result, where $R^*$ denotes the set of \emph{exact} $k$ nearest neighbors as defined in \defref{def:nngr}.

\subsection{NSW for ANN Search}
\label{subsec:hnsw}

A mainstream approach to approximate nearest neighbor search employs navigation small world (NSW) graphs ~\cite{vldb19nsw,pami20hnsw}.
In an NSW graph $G^*=(V^*,E^*)$, each node $v^*_i$ corresponds to a vector $\mathbf{e}_i$ and is connected to a small set of neighboring nodes whose vectors are highly similar ({\ie} close under the distance metric $d(\cdot,\cdot)$), thereby forming a ``small world'' neighborhood around the node.
Given a query vector $\mathbf{e}_q$, an ANN search can be performed using the precedure outlined in \algref{alg:nsw-search}.
Notably, the entry point $ep$ could be an arbitrary point in the classic NSW~\cite{vldb19nsw} or a specially selected node in Hierarchical NSW (HNSW)~\cite{pami20hnsw}.

\begin{algorithm}[t!]

    \KwIn {NSW graph $G^*$, query vector $q$, entry point $ep$, predicate function $f$, beam search width $b$, $k$ nearest neighbors.}
    mark $ep$ as visited; \\
    initialize min-heap $pool$ and max-heap $ann$; \\
    push $ep$ to $pool$ and $ann$; \\
    \While{$pool$ is not empty}{
        $u\leftarrow$nearest vector to $q$ in $pool$; $pool.pop()$; \\
        $v\leftarrow$farthest vector to $q$ in $ann$; $ann.pop()$; \\
         \textbf{if}{$dis(q,u)>dis(q,v)$ and $|ann|=b$} \textbf{then break;}
         \ForEach{unvisited $o\in G[u]$}{
            mark $o$ as visited and add $o$ to $pool$; \\
            \If{$in\_range(v_q,v_o,r)$ and $(|ann|<b$ or $dis(q,o)<dis(q,v))$}{
                add $o$ to $ann$;
            }
         }
    }
    $ann\leftarrow$ $k$ nearest vectors to $q$ in $ann$; \\
    \Return{$ann$};
    
 \caption{ANN Search with NSW Graph}
 \label{alg:nsw-search}
 \vspace{-2pt}
\end{algorithm}

In \algref{alg:nsw-search}, a min-heap $pool$ and a max-heap $ann$ are maintained to store candidate nodes and filtered results, respectively (lines 2-3).
The algorithm iteratively pops the vector $u$ with the smallest distance to the query vector $\mathbf{e}_q$ from $pool$ (line $5$) and evaluates its neighbors in $G^*$.
For each neighbor that satisfies the predicate function $in\_range$, the algorithm inserts it into $ann$, provided $ann$ has fewer than $k$ elements or the new point is closer to $\mathbf{e}_q$ than the farthest point currently in $ann$ (lines 8-10).
The search terminates when all remaining candidates in $pool$ are no closer to $\mathbf{e}_q$ than the farthest vector in $ann$.
Notably, while the graph traversal considers all reachable candidates, only those satisfying the predicate $in\_range$ are admitted into $ann$.

\stitle{In-Range Query for ANNGR Problem.}
For the ANNGR problem, the predicate function $in\_range(v_q,v_o,r)$ checks whether the node $v_o$ (associated with vector $\mathbf{e}_o$ in the filter graph $G$) lies within $r$ shortest-path distance of the node $v_q$ (associated with query vector $\mathbf{e}_q$).
We define this check as \emph{in-range query}.

\begin{definition}[In-Range Query]
    Given the query node $v_q$ and a candidate node $v_o$ retrieved from the filter graph $G$, an in-range query determines whether $v_o$ lies within graph range $r$ of $v_q$ in $G$, {\ie} whether $d_G(v_q, v_o) \leq r$.
\end{definition}

The in-range query can be achieved in various ways. For instance, a straightforward approach is to first perform a breadth-first search (BFS) from $v_q$ in the filter graph $G$ to enumerate all nodes within graph distance $r$, and then check whether the candidate node $v_o$ belongs to this set.
Other shortest-path indexing structures can also be employed for the in-range query, as illustrated in \secref{subsec:hnsw-label}.

The search procedure outlined in \algref{alg:nsw-search}, combined with the in-range query, consists of a \emph{post-filtering} solution~\cite{corr25filterexp} to the ANNGR problem.
Candidate vectors close to the query vectors are first retrieved and then decided whether they are located in the $r$-hop of the query node in the filter graph.

\subsection{Naive Solution: NSW with Labeling}
\label{subsec:hnsw-label}
In the ANNGR problem, evaluating the predicate requires computing the shortest-path distance between the query node $v_q$ and the candidate node $v_o$ to determine whether $d_G(v_q, v_o)\leq r$.
This motivates the use of shortest path indexing techniques to enable efficient in-range queries.
Recent works has proposed Pruned Landmarked Labeling (PLL)~\cite{sigmod13pll}, which supports fast and exact shortest-path distance queries between node pairs in a graph.

\begin{definition}[Pruned Landmarked Labeling (PLL)]
\label{def:2hoplabel}
    Given a graph $G=(V,E)$, a Pruned Landmarked Labeling (PLL) consists of a labeling $L$ that assigns to each node $v\in V$ a set of node-distance pairs $L(v) = \{(u, d_G(u, v)) \mid u \in C(v)\}$, where $C(v)\subseteq V$ is the set of \emph{hub nodes} of $v$.
    The labeling $L$ satisfies the following 2-hop cover property: for any pair of nodes $u,v\in V$, there exists at least one node $w\in C(u)\cap C(v)$ such that $d_G(u,v)=d_G(u,w)+d_G(w,v)$.
    In other words, the shortest path between $u$ and $v$ can be reconstructed via a common hub $w$ in their label sets.

\end{definition}

\begin{figure}[t!]
\centering

\includegraphics[height=0.39
\textwidth]{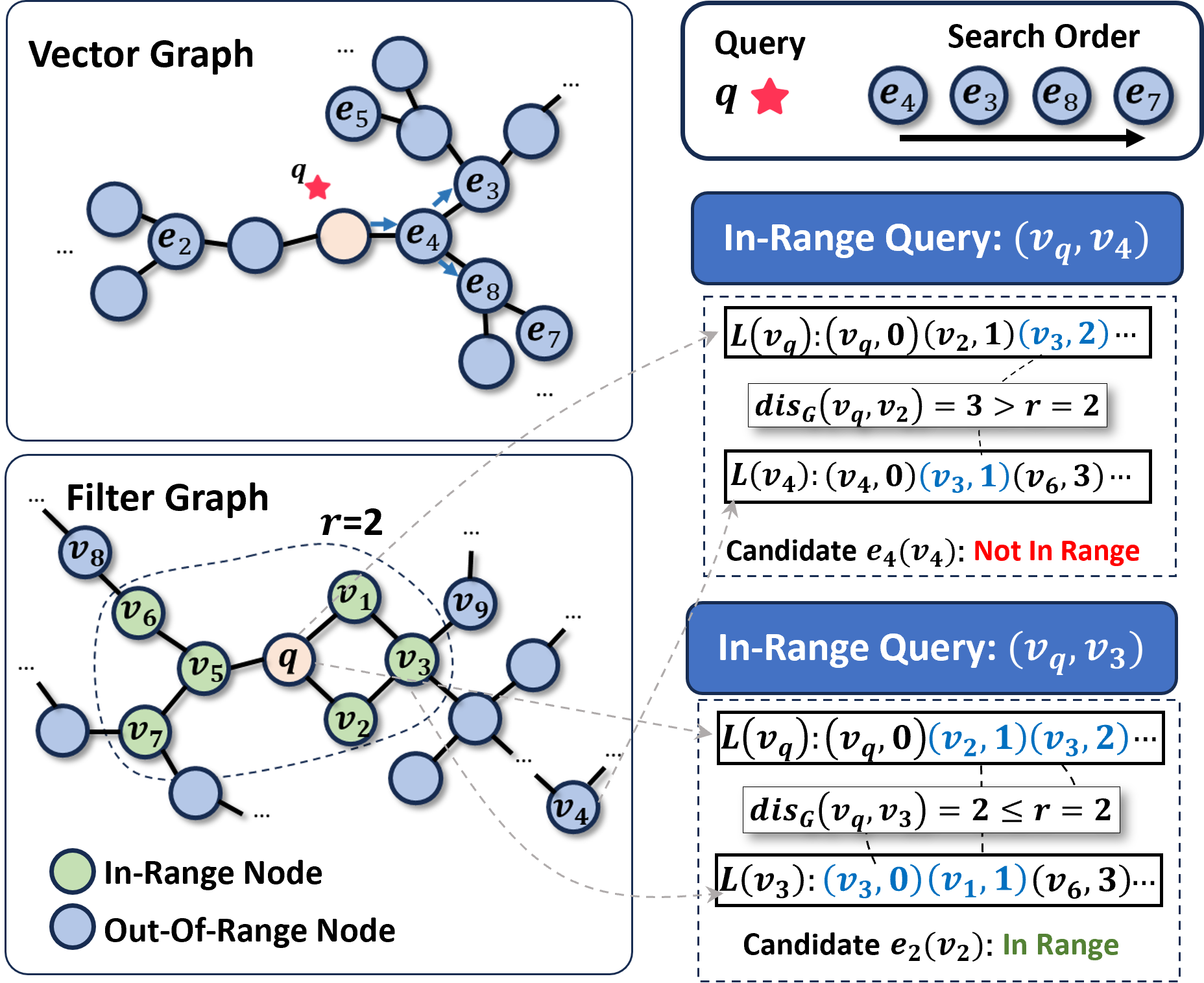}
\caption{An example for Labeling NSW.}

\label{fig:eg-labelnsw}
\end{figure}

Given the PLL of a graph $G$, the shortest-path distance between two nodes $u,v\in V$ can be computed by examining the paths that pass through their common hub nodes in $C(u)\cap C(v)$.
Formally, the distance is given by
\begin{align}\label{eq:2hl-dis}
    dis(u,v)=\min_{y\in C(u)\cap C(v)}dis(u,y)+dis(y,v).
\end{align}
We refer the reader to previous works~\cite{sigmod13pll} for details on the correctness and construction of PLL.

\stitle{Labeling NSW for ANNGR Problem}.
To address the ANNGR problem, a straightforward approach is to associate each vector $\mathbf{e}_i$ in the NSW index with the computed PLL $L(v_i)$ of its corresponding node $v_i$ in the filter graph $G$, which enables a possible solution to the in-range filtering process during search.
As illustrated in \figref{fig:eg-labelnsw}, when evaluating whether the node $v_o$ of a candidate vector $\mathbf{e}_o$ lies within $r$ hops of the query node $v_q$ ({\ie} at line 9 of \algref{alg:nsw-search}), we compute the shortest-path distance $d_G(v_q, v_o)$ using the PLL via Equation~\eqref{eq:2hl-dis}, and accept $\mathbf{e}_o$ if $d_G(v_i, v_q) \leq r$.


\stitle{Limitations of Labeling NSW.}
Based on the correctness guarantees of PLL~\cite{sigmod13pll}, Labeling NSW can always accurately determine whether a candidate node lies within $r$ hops of the query node $v_q$ in the filter graph.
However, in practice, the space overhead of the labeling is substantial: the size of PLL often grows superlinearly with the scale and diameter of the graph, leading to significant memory consumption and cache inefficiency.  
This overhead degrades query throughput to the extent that Labeling NSW underperforms even a naive baseline that explicitly materializes all $r$-hop neighbors of $v_q$ during query processing.

\subsection{Bloom Filter}
\label{subsec:bf}
A Bloom filter~\cite{acm70bf,sur19bf} is a space-efficient probabilistic data structure for approximate set membership queries.
Formally, a Bloom filter consists of an array $B$ of $m$ bits (initially all set to $0$) and a family of $t$ independent hash functions $H=\{h_1,...,h_t\}$, each mapping elements uniformly to $\{0, 1, \dots, m-1\}$.
The best false positive probability (FPP) of a Bloom filter is $\alpha=(\frac{1}{2})^{\frac{m}{r}ln2}$ where $r$ is the number of inserted elements.

\begin{algorithm}[t!]

    \SetKwFunction{funcinsert}{insert}
    \SetKwFunction{funccontain}{contain}
    \SetKwFunction{funcaic}{aic}
    
    \SetKwProg{Fn}{Function}{:}{end}
    \Fn{\funcinsert{$B$, $b$}} {
    \ForEach{$h_i\in H$}{
    $B[h_i(b)]\leftarrow 1$;
    }

    }
    \Fn{\funccontain{$B$, $b$}}{
    \ForEach{$h_i\in H$} {
    \textbf{if} $B[h_i(b)]=0$ \textbf{do} \Return{false};
    }
    \Return true;
    }

    \Fn{\funcaic{$B_1$, $B_2$}} {
        $c\leftarrow$0; \\
        \ForEach{$i=0,...,m$}{
            $c\leftarrow c+B_1[i]|B_2[i]$;
        }
        $ac\leftarrow B_1.cnt+B_2.cnt+\frac{m}{k}\ln{\{1-\frac{ac}{m}\}}$; \\
    }
 \caption{Bloom Filter Operations}
 \label{alg:bf-op}
 \vspace{-2pt}
\end{algorithm}

As shown in \algref{alg:bf-op}, Bloom filters support three key operations: \emph{insertion}, \emph{containment query}, and \emph{approximate estimation of set intersection}.
To \emph{insert} an element $b$, all bits $B[h_1(b)]$, \dots, $B[h_t(b)]$ are set to $1$ (lines 2-3).
To query whether the Bloom filter \emph{contains} $b$, we check if all bits $B[h_1(b)],\dots,B=h_t[(b)]$ are set to $1$.
If any bit is $0$, $b$ is definitely not in the set (line 6).
However, if all bits are $1$ (line 7), the query may yield a \emph{false positive}, {\ie} $b$ might not actually be in the set.

For two Bloom filters $B_1$ and $B_2$ built over the same parameters, the size of their set intersection, {\ie} the \texttt{aic} function in \algref{alg:bf-op}, can be approximated using the number of $1$-bits from the bitwise OR operations.
Specifically, let $w=B_1 \lor B_2$ denote the bitwise OR cardinality ({\ie} the count of positions where either filter has a $1$).
The intersection size could then be estimated via the inclusion-exclusion principle, as illustrated in line 12 of \algref{alg:bf-op}.
\begin{figure*}[t!]
\centering

\includegraphics[height=0.41\textwidth]{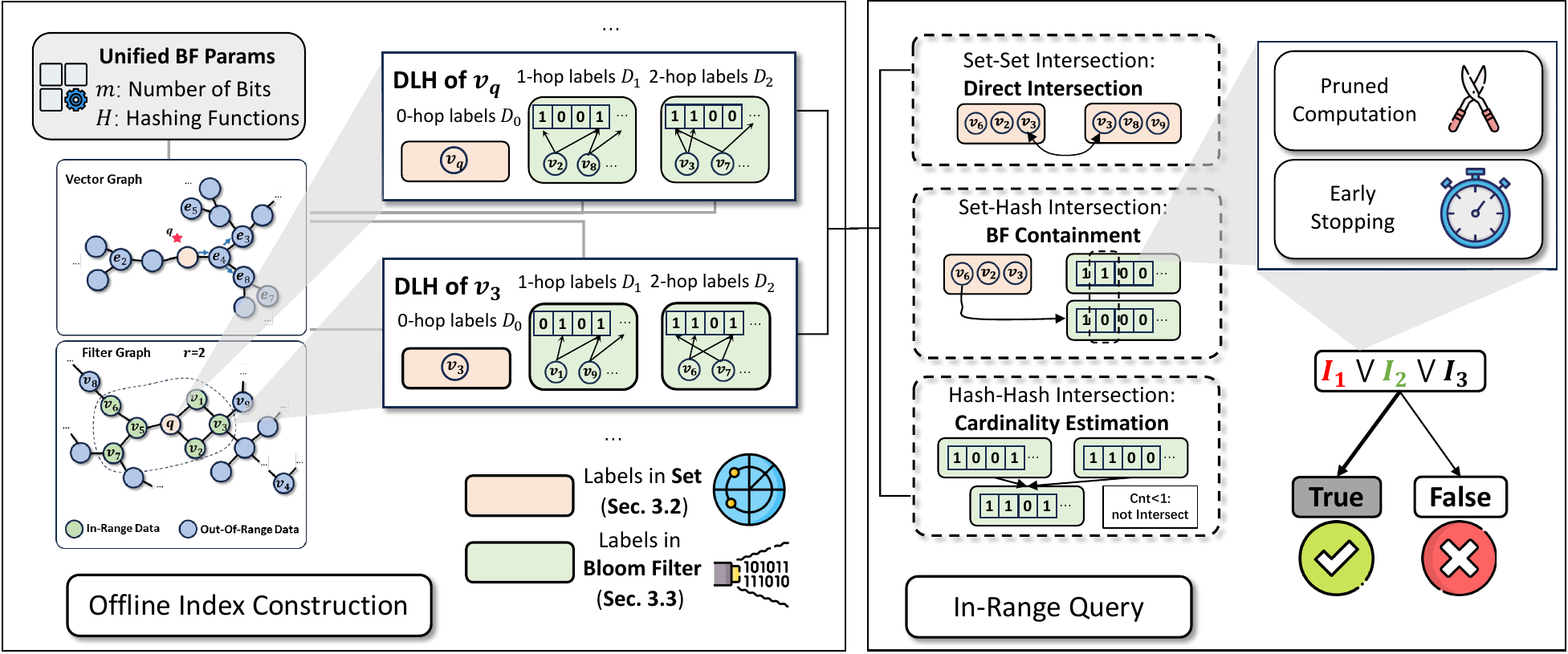}
\caption{Overview of {\ourindex}.}
\label{fig:overview}
\end{figure*}

\section{{\ourindex}: A Post-Filtering Approach}
\label{sec:alg}

This section first presents an overview of {\ourindex}, followed by detailed descriptions of its two core components: Distance-Aware Labeling and Hash-Based Compression.

\subsection{Overview}
\label{subsec:overview}

Figure~\ref{fig:overview} presents the  overview of the proposed {\ourindex}.  
In general, {\ourindex} consists of two main components: Distance-Aware Labeling introduced in \secref{subsec:dal}, which enables efficient and pruned in-range queries for the ANNGR problem, and Hash-Based Compression introduced in \secref{subsec:hash-label}, which further supports fast approximate in-range queries using Bloom filters.

\stitle{Distance-Aware Labeling.}
For each node in the \emph{filter graph}, {\ourindex} stores partial labels corresponding to nodes whose distance from the source is at most the graph range $r$. These labels are partitioned into distinct labeling sets based on their distance from the source node (shown in orange in Figure~\ref{fig:overview}), a strategy referred to as Distance-Aware Labeling (DAL).  
DAL not only enables efficient in-range queries through simple set intersections, but also contains  \emph{distance-aware pruning} and \emph{early stopping} strategies that eliminate unnecessary computations.

\stitle{Hash-Based Compression.}
To further enhance efficiency and reduce the index size, {\ourindex} compresses each labeling set containing more than $W$ elements into a Bloom filter (shown in green in Figure~\ref{fig:overview}).  
Under this hashing-based transformation, set intersections can be achieved by containment queries and approximate intersection cardinality estimates using Bloom filters, whenever one or both of the input labeling sets are represented in Bloom filter.

As shown in the right part of Figure~\ref{fig:overview}, under {\ourindex}, an \emph{in-range query} is answered by performing set intersection operations between the distance-aware labeling sets of the two query nodes. The specific manner of this operation depends on the representation of the labeling sets ({\ie} explicit sets or Bloom filters).  
Furthermore, thanks to the design of Distance-Aware Labeling in {\ourindex}, unnecessary set intersections can be pruned early, and an early stopping mechanism is employed to further enhance efficiency.

\subsection{Distance-Aware Labeling: In-Range Query with Pruned Set Intersection}
\label{subsec:dal}

\stitle{Rationale.}
In the in-range queries of the ANNGR problem, we do not need to know the exact shortest path length between the query node $v_q$ and a candidate node $v_o$.
Instead, we only need to determine \emph{whether the shortest path between $v_q$ and $v_o$ is at most $r$}.  
This observation motivates the design of a Distance-Aware Labeling (DAL) index, which partitions the labeling hubs by distance to the source node, to enable an efficient solution to the ANNGR problem.

\stitle{Construction of Distance-Aware Labeling.}
For each node $v$ in $G$, we partition its PLL labeling $L(v)$ into multiple labeling sets based on the distance from each labeled node to the source node (i.e., $v$ itself).  
Consequently, during an in-range query, determining whether a candidate node $v_o$ lies within distance $r$ of the query node $v_q$ reduces to performing set intersection operations between the corresponding Distance-Aware Labeling sets of $v_q$ and $v_o$, enabling efficient in-range verification for the ANNGR problem.

Formally, we define the $i$-path 
labeling set, denoted by 
\begin{equation}
    D_i(v)=\{u|(u,dis_G(u,v))\in L(v) \land dis_G(u,v)=i\}
\end{equation}
as the set of nodes $u$ such that (i) $u$ appears in $L(v)$, and (ii) the shortest-path distance from $u$ to $v$ is exactly $i$.
The \emph{Distance-Aware Labeling} of $v$ is then defined as the collection of $i$-path labeling sets for $i\leq r$, denoted by 
\begin{equation}
D(v)=\{D_0(v),...,D_r(v)\}.
\end{equation}
If the context is clear, we let $d_v^i=|D_i(v)|$ denote the size of the set, and define $d_v=\sum_{i\leq r}d_v^i$ as the total size of all labeling sets in $D(v)$.

\begin{algorithm}[t!]

    \KwIn {The PLL index $L(v)$ of $v$.}
    $D_i(v)\leftarrow\emptyset$ for $0\leq i\leq r$; \\
    \For{$(u,dis_G(u,v))\in L(v)$} {
        \If{$dis_G(u,v)\leq r$} {
            Append $u$ to $D_{dis_G(u,v)}(v)$;
        }
    }
    \Return{$\{D_0(v),...,D_r(v)\}$};
    
 \caption{Distance-Aware Labeling Construction}
 \label{alg:dal-const}
\end{algorithm}

\algref{alg:dal-const} presents the construction of distance-aware labeling from the shortest path labeling index~\cite{sigmod13pll}.
Only $i$-path labeling sets with $i \leq r$ are stored.  
To construct the Distance-Aware Labeling for a node $v$, we simply iterate over the entries in $L(v)$ (line~2) and append each qualifying node to the corresponding labeling set $D_i(v)$ (lines~3–4).
Labels corresponding to distances greater than $r$ are discarded, as they do not affect the results of in-range queries.

\stitle{In-Range Query with Distance-Aware Labeling.}
In contrast to shortest-path queries on traditional shortest path labeling index, which require computing the exact distance, we only need to determine whether a path of length at most $r$ exists in the filter graph.
This observation motivates two key pruning strategies based on distance-aware labeling to enable efficient in-range queries.

\noindent\emph{\underline{Pruning Computations Based on Distance.}}  
Certain computations can be pruned by leveraging distance information.  
Although we retain all label entries whose distance to the query node $v_q$ is at most $r$, the intersection between two DAL labeling sets $D_i(v_q)$ and $D_j(v_o)$ can be safely skipped whenever $i+j>r$, since any path through a common hub would have length greater than $r$.

\noindent\emph{\underline{Early Stop for In-Range Query.}}
In the ANNGR problem, the in-range query seeks only to determine whether the shortest path between $v_q$ and $v_o$ is at most $r$, rather than computing its exact length.  
This implies that we can immediately conclude $v_o$ lies within $r$-range of $v_q$ as soon as a common hub yielding a path of length $\leq r$ is found, without enumerating all entries in $D(v_q)$ and $D(v_o)$.

Algorithm~\ref{alg:dal-query} presents the pseudocode for the in-range query based on Distance-Aware Labeling.  
Since the labels are partitioned by their distance and the query range $r$ is pre-defined, the algorithm iterates over all pairs of indices $(i, j)$ such that $i + j \leq r$ (lines~1–2) and checks whether the labeling sets $D_i(v_q)$ and $D_j(v_o)$ intersect.  
If an intersection is found, it implies the existence of a path of length at most $r$ between $v_q$ and $v_o$, and the algorithm immediately returns the result \texttt{true}.
Otherwise, the algorithm returns the result \texttt{false} if no intersection is found (line 5).

\stitle{Addressing ANNGR with DAL.}
Leveraging the Distance-Aware Labeling design, the ANNGR problem is solved by integrating Algorithm~\ref{alg:nsw-search} and Algorithm~\ref{alg:dal-query}.  
Specifically, the DAL index is constructed offline from the filter graph prior to query processing (as illustrated in Algorithm~\ref{alg:dal-const}).  
Given a query $q$, Algorithm~\ref{alg:nsw-search} is used to enumerate candidate vectors and select the results whose corresponding nodes satisfy the graph range filters in a post-filtering manner.  
Whenever a candidate $o$ requires verification, i.e., whether the corresponding node $v_o$ lies within distance $r$ of the query node $v_q$ (as in the in-range query at line~9 of Algorithm~\ref{alg:nsw-search}), Algorithm~\ref{alg:dal-query} is invoked to perform this check using the DAL index.

\begin{algorithm}[t!]

    \KwIn {Node of query vector $v_q$ along with labeling $D(v_q)$, a candidate node $v_o$ along with labeling $D(v_o)$, range $r$.}

    \For{$i=0,...,r$}{
        \For{$j=0,...,r-i$}{
            \If{$D_i(v_q)\cap D_j(v_o)\neq \emptyset$}{
            \Return{true};
        }
    }
    }
    \Return{false};
    
 \caption{In-Range Query with Distance-Aware Labeling for ANNGR problem}
 \label{alg:dal-query}
\end{algorithm}

\begin{example}
\label{eg:dal}
Figure~\ref{fig:eg-dal} illustrates an example of the construction and query process using Distance-Aware Labeling.  
During the offline construction phase, labels are partitioned into sets according to their distance from the source node.  
Labels with distance greater than the query range $r = 3$ (i.e., the gray crossed-out entries) are discarded.  
At query phase, suppose we need to determine whether a candidate node $v_o$ lies within distance $r$ of the query node $v_q$.  
We then evaluate set intersections between $D_i(v_q)$ and $D_j(v_o)$ for all pairs $(i, j)$ such that $i + j \leq r$ (indicated by the black arrows with marks \ding{172} and \ding{173} in the bottom part of Figure~\ref{fig:eg-dal}). 
The intersections with $i+j>r$ will be discarded (indicated by the gray arrow with mark \ding{175} in Figure~\ref{fig:eg-dal}).
\end{example}

\begin{figure}[t!]
\centering

\includegraphics[height=0.31\textwidth]{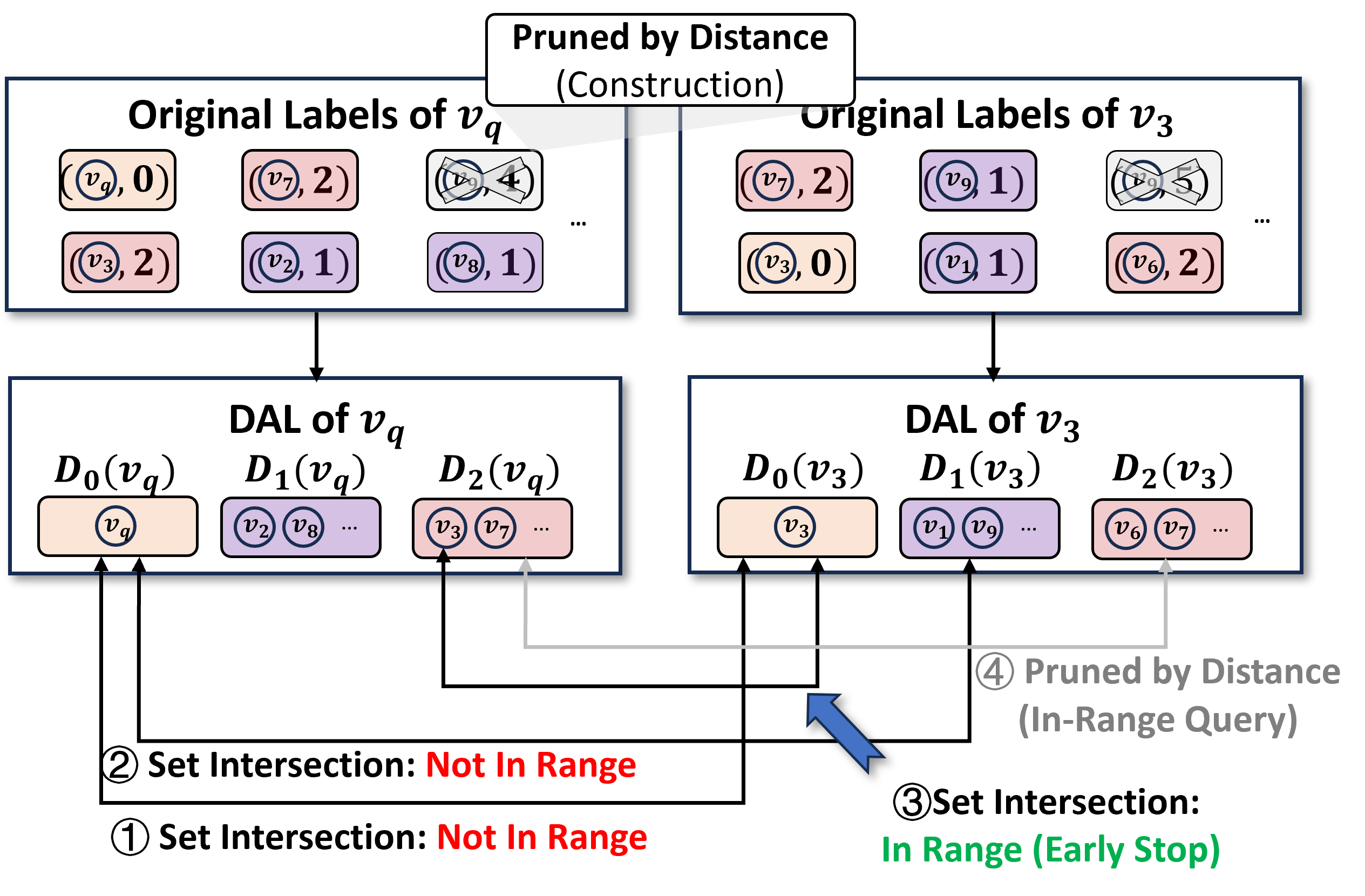}
\caption{Example of Distance-Aware Labeling.}
\label{fig:eg-dal}

\end{figure}

\stitle{Time Complexity.}
The \emph{construction of the DAL} incurs a time complexity of \( O(\sum_{v\in G} |L(v)|) \), as it requires enumerating all entries in the shortest-path index.
For \emph{in-range query} using DAL, computing the intersection between $D_i(v_q)$ and $D_j(v_o)$ takes $O(|D_i(v_q)|+|D_j(v)|)$, when hash tables are used to represent sets.
Consequently, the total time complexity of \algref{alg:dal-query} is
\begin{align}
    \sum_{i+j\leq r}O(|D_i(v_q)|+|D_j(v_o)|)=O(d_{v_q}+d_{v_o})
\end{align}
where $d_{v_q}$ and $d_{v_o}$ denote the total sizes of the labeling sets in $D(v_q)$ and $D(v_o)$, respectively.

Building on the correctness of PLL, which always returns the exact shortest-path length between two nodes, the proposed Distance-Aware Labeling guarantees the correctness of DAL.

\begin{corollary}
\label{cor:dal-correct}
    Distance-Aware Labeling, together with \algref{alg:dal-query}, correctly determines whether a candidate node $v_o$ lies within the $r$-hop of the query node $v_q$.
\end{corollary}

\corref{cor:dal-correct} holds because the shortest-path distance $r'$ between nodes $v_q$ and $v_o$ satisfies $r'\leq r$ if and only if there exist labeling sets $D_i(v_q)$ and $D_j(v_o)$ such that $i+j\leq r$ and $D_i(v_q)\cap D_j(v_o)\neq \emptyset$.
Consequently, any pair of sets $D_i(v_q)$ and $D_j(v_o)$ with $i+j>r$ can be safely pruned, leading to reduced pruned computation overhead.
Moreover, as long as a non-empty intersection is found, the in-range query can immediately return \texttt{true}, which validates the correctness of the early stop strategy.

\subsection{DLH: Accelerating DAL via Hashing}
\label{subsec:hash-label}

\eat{\secref{subsec:dal}, although those labeling nodes with a limited distance to the source node are maintained in Distance-Aware Labeling, the size of some labeling sets is still large, especially for those with a large distance. }

\stitle{Rationale.}
The employment of hashing techniques is motivated by the following key observations.
(1) \emph{ANN search allows false positives.}
Approximate nearest neighbor search permits a controlled relaxation of correctness by allowing false positives.  
Motivated by this trade-off, hashing techniques can be leveraged to accelerate query processing for the ANNGR problem further.
(2) \emph{The employed technique should support approximate set intersection.}
DAL relies on set intersection operations for in-range queries in the ANNGR problem.
The hash-based representations of the labeling sets must facilitate efficient set intersection operations to enable in-range queries.
To this end, this paper adopts the \emph{Bloom filter} as a compact, compressed representation of labeling sets.

\stitle{Bloom Filters for Efficient Set Intersections.}
Given a set of predefined hash functions $H$, each $i$-path labeling set $D_i(v)$ whose size exceeds a threshold $T$ is mapped to a Bloom filter of bit length $m$, denoted by $B_i(v)$.
During the set intersection checks in \algref{alg:dal-query}, whenever one or both of the involved labeling sets  have corresponding Bloom filters, we instead perform set intersection using these Bloom filters.
For completeness, we define $B_i(v)=\emptyset$ if no Bloom filter is constructed for $D_i(v)$, and let $B(v)=\{B_1(v),...,B_r(v)\}$ denote the collection of all Bloom filters associated with node $v$.
 
\begin{algorithm}[t!]

    \KwIn {Node of query vector $v_q$ along with labeling $D(v_q)$ and bloom filter sets $B(v_q)$, node of a candidate $v_o$ along with labeling $D(v_o)$ and Bloom filter sets $B(v)$, range $r$.}

    \For{$i=0,...,r$}{
        \For{$j=0,...,r-i$}{
        \If{$B_i(v_q)\neq\emptyset$ \textbf{and} $B_i(v_o)\neq\emptyset$} {
            \If{\texttt{aic}$(B_i(v_q),B_j(v_o))\geq 1$}{
            \Return{true};
            }
        }
        \ElseIf{$B_i(v_q)\neq \emptyset$}{
            \ForEach{$b\in D_j(v_o)$}{
            \If{\texttt{contain}$(B_i(v_q),b)$}{
            \Return{true};
            }
            }
        }
        \ElseIf{$B_j(v_o)\neq\emptyset$}{
        \ForEach{$b\in D_i(v_q)$}{
        \If{\texttt{contain}$(B_j(v_o),b)$}{
        \Return{true};
        }
        }
        }
        \Else{
            \If{$D_i(v_q)\cap D_j(v_o)\neq \emptyset$}{
            \Return{true};
            }
        }
    }
    }
    \Return{false};
    
 \caption{In-Range Query With DLH}
 \label{alg:dal-bf-query}
\end{algorithm}

\algref{alg:dal-bf-query} presents the pseudocode for the query procedure in {\ourindex}.
{\ourindex} retains the distance-aware strategy for in-range queries.
When determining whether the intersection of the $i$-path labeling set of $v_q$ and the $j$-path labeling set of $v_o$ is empty, the algorithm selects different intersection-checking operations depending on the availability of Bloom filters.
(1) If both labeling sets are equipped with Bloom filters (lines 3-5), an approximate estimation of set intersection is computed.
The two sets are considered to intersect if this count exceeds zero.
(2) If only one set has a Bloom filter (lines 6-13), the algorithm iterates over each element in the other set and checks its containment in the Bloom filter.
(3) Finally, when neither set has an associated Bloom filter (lines 14-16), the algorithm falls back to the exact intersection method used in \algref{alg:dal-query}.

\begin{example}

    Consider the labeling sets in \egref{eg:dal}.
    \figref{fig:eg-dalha} illustrates how the {\ourindex} index is derived from the DAL index and employed during in-range queries.
    Specifically, for the node $v_q$ in the filter graph, its labeling sets, namely $D_1(v_q)$ and $D_2(v_q)$, both exceed the size threshold  $T=3$ and are therefore encoded into Bloom filters.
    Likewise, $D_1(v_3)$ and $D_2(v_3)$ of $v_3$ are also converted into Bloom filters.
    During the in-range query between $v_q$ and $v_3$, the algorithm selects the appropriate intersection strategy based on the representation types of the two labeling sets:
    \begin{itemize}
        \item Set–set intersection is used  for the intersection between $D_0(v_q)$ and $D_0(v_3)$;
        \item Set–hash intersection is used for the intersection between a labeling set and a Bloom filter, like $D_0(v_q)$ and $B_2(v_3)$, and $B_2(v_q)$ and $D_0(v_3)$;
        \item Hash–hash intersection is used for the intersection between $B_1(v_q)$ and $B_1(v_3)$.
    \end{itemize}
    Specifically, since the intersection between $B_2(v_q)$ and $D_0(v_3)$ is non-empty, the in-range query terminates early and returns \texttt{true}.
\end{example}

\stitle{Time Complexity.}
Let $\mathcal{D}_<(v)$ denote the collection of distance-aware labeling sets for node $v$ whose sizes are at most $T$, {\ie} those stored explicitly as sets, and let $d_v^<$ be the total size of all sets in $\mathcal{D}_<(v)$.
Formally,
\begin{align}
    \mathcal{D}_<(v)=\{D_i(v)|\ |D_i(v)|\leq T\}, d_v^<=\sum_{D_j(v)\in\mathcal{D}_<(v)}|D_j(v)|.
\end{align}
We then obtain the following lemma.
\begin{lemma}
\label{lem:time-comp}
    The time complexity of \algref{alg:dal-bf-query} is $O(d_{v_q}^<\cdot t+d_v^<\cdot t+c_{v_q}^>\cdot c_v^>\cdot t)$, where $c_v^>$ denotes the number of labeling sets whose size is at least $T$ and $t$ is the number of hashing functions in the Bloom filters.
    Further considering that $d_{v_q}^<,d_v^<<rT$ and $c_{v_q}^>,c_v^>\leq r$, the time complexity is bounded by $O(rtT+r^2t)$.
\end{lemma}
\begin{proof}

\begin{figure}[t!]
\centering

\includegraphics[height=0.27\textwidth]{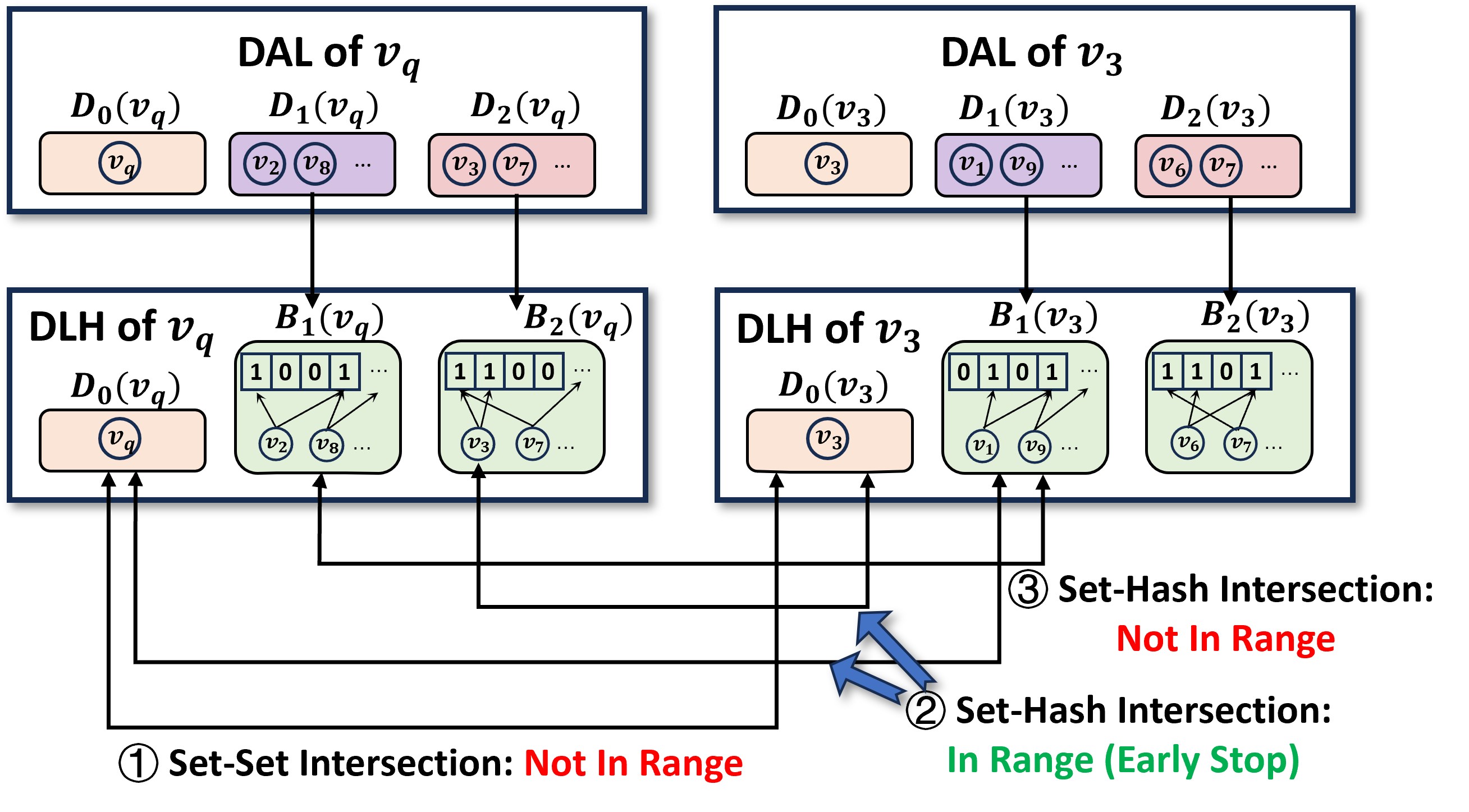}
\caption{Example of DLH.}
\label{fig:eg-dalha}
\vspace{-6pt}
\end{figure}

We analyze the time complexity by case distinction.

\noindent
(1) For set intersections between two sets in $\mathcal{D}_<(v)$, our analysis in \secref{subsec:dal} shows the total time complexity is $O(d_{v_q}^<+d_v^<)$.

\noindent 
(2) Consider the set intersection between a labeling of $D_i(v_q)$ in explicit \emph{set} form and another labeling set $D_j(v)$ represented as a \emph{Bloom filter} $B_j(v)$.
To determine whether $B_j(v)$ contains any element from $D_i(v_q)$, 
each element from $D_i(v_q)$ is evaluated against all $t$ hash functions of the Bloom filter (see \algref{alg:bf-op}), which takes \( O(t) \) time per element.
Since this check is performed for every element in all such set-type labeling sets of \( v_q \), the total time complexity for this case is
\begin{align}
    &\sum_{D_i(v_q)\in\mathcal{D}_<(v_q)}\sum_{e\in D_i(v_q)}t = d^<_{v_q}t.
\end{align}

\noindent
(3) Similarly to case (2), for the intersection between a Bloom filter $B_i(v_q)$ and a labeling set $D_j(v)$ stored in explicit \emph{set} form, the total time complexity is $d^<_{v}t$.

\noindent
(4) For the intersection between two labeling sets represented as \emph{Bloom filters}, {\ie} $B_i(v_q)$ and $B_j(v)$, 
the approximate intersection count is typically computed using bit-level operations, with the time complexity $O(t)$.
This comparison is performed for all pairs of labeling sets that are mapped to Bloom filters, yielding a total time complexity of $O(c^>_{v_q}\cdot c^>_{v}\cdot t)$.

Thus, the total time complexity is $O(d_{v_q}^<\cdot t+d_v^<\cdot t+c_{v_q}^>\cdot c_v^>\cdot t)$.
Besides, since there are at most $r$ sets in $\mathcal{D}_<(v)$ and each set has size $<T$, we have $d_{v_q}^<$ and $d_v^<<rT$.
The time complexity is $O(rtT+r^2t)$.

\end{proof}

\eat{
\stitle{How {\ourindex} affects performance of ANN-GR?}
Due to the false positivity of bloom filters, some out-of-range vectors may be mistakenly considered as an in-range vectors.
This part theoretically proves how the recall changes under {\ourindex}.

\begin{theorem}
\label{thm:recall}
    123
\end{theorem}
\begin{proof}
    Consider the set intersection between a \emph{set} $D_i(v_1)$ and a \emph{bloom filter} $B_j(v_2)$.
    Define $Err_1(D,B)$ the event that the intersection between the set $D$ and the bloom filter $B$ emerges false positivity.
    From \todo{cite}, the probability of a bloom filter $B$ making mistake in one test is 
    \begin{align}
        Pr[Err_1(\{l\},B_j(v))]=[1-(1-\frac{1}{m})^{kn}]^k\leq [1-(1-\frac{kn}{m})]^k=\frac{k^kn^k}{m^k}
    \end{align}
    where $k$ and $m$ are the number of hash functions and bits in the bloom filter, and $n$ is the number of inserted elements.
    Therefore, the probability of a false positive caused by the intersection between a set and a bloom filter $B_j(v_2)$ is
    \begin{align}
        Pr[ Err_1(\mathcal{D}_<(v_q),B_j(v)]
        =& Pr[\cup_{D_i(v_q)\in \mathcal{D}_<(v_q)} Err_1(D_i(v_q),B_j(v)]\notag \\
        \leq &\sum_{ D_i(v_q)\in\mathcal{D}_<(v)}Pr[Err_1(D_i(v_q),B_j(v))]\notag \\
        \leq &\sum_{ D_i(v_q)\in\mathcal{D}_<(v)}\sum_{l\in D_i(v_q)} Pr[Err_1(\{l\},B_j(v)] \notag \\
        \leq &d_<(v_q)[1-(1-\frac{1}{m})^{kn}]^k.
    \end{align}
Similarly, the probability of a false positive caused by the intersection between a set and the bloom filter $B_i(v_q)$ is upper bounded by $Pr[Err_1(\mathcal{D}_<(v),B_i(v_q)]=d_<(v)[1-(1-\frac{1}{m})^{kn}]^k$.

For $Err_2(B_i(v_q),B_j(v))$, the event that the intersection between two bloom filters results in a false positive.
Since the intersection is derived from the union of two bloom filters, it is equivalent to the event that \emph{$(B_i(v_q).cnt+B_j(v).cnt)$ elements are inserted to the bloom filter and the $\texttt{aic}$ function returns an result $>1$}.
This means we need the number of triggered bits in the bloom filter after inserting $bs=B_i(v_q).cnt+B_j(v).cnt$ elements, {\ie} $c$ in \algref{alg:bf-op}, satisfies
\begin{align}
    -\frac{m}{k}\ln\{1-\frac{c}{m}\}\leq bs-1 \Leftrightarrow c\geq m(1-e^{\frac{(1-bs)k}{m}})
\end{align}

    \todo{balls in bins}
\end{proof}
}

\begin{algorithm}[t!]

    \KwIn {NSW graph $G^*$, query vector $q$, entry point $ep$, predicate function $f$, beam search width $b$, $k$ nearest neighbors.}
    mark $ep$ as visited; \\
    initialize min-heap $pool$ and max-heap $ann$; \\
    push $ep$ to $pool$ and $ann$; \\
    {\color{blue}{$P_L\leftarrow \texttt{preh}(D(v_q))$}}; \\
    \While{$pool$ is not empty}{
        $u\leftarrow$nearest vector to $q$ in $pool$; $pool.pop()$; \\
        $v\leftarrow$farthest vector to $q$ in $ann$; $ann.pop()$; \\
         \textbf{if}{$dis(q,u)>dis(q,v)$ and $|ann|=b$} \textbf{then break;}
         \ForEach{unvisited $o\in G[u]$}{
            mark $o$ as visited and add $o$ to $pool$; \\
            \If{{\color{blue}{$in\_range\_pre(P_L,v_o,r)$}} and $(|ann|<b$ or $dis(q,o)<dis(q,v))$}{
                add $o$ to $ann$;
            }
         }
    }
    $ann\leftarrow$ $k$ nearest vectors to $q$ in $ann$; \\
    \Return{$ann$};
    
 \caption{ANN Search for DLH-M}
 \label{alg:batch-nsw-search}
 \vspace{-2pt}
\end{algorithm}

\section{{\ourindex}-M: {\ourindex} With Memoization}
\label{sec:par}

\secref{sec:alg} presents a post-filtering approach for ANNGR problem based on the {\ourindex} index.
In this section, we observe that range queries in {\ourindex} exhibit asymmetry.
This observation motivates a more efficient solution to the ANNGR problem, achieved by the memoization of intermediate hashing results prior to the ANN search process.

\stitle{Rationale.}
In {\ourindex}, in-range queries always focus on determining whether a candidate node lies within the range of the \emph{query node}. 
For set intersections between the explicit labeling sets of the query node and the Bloom filters of candidate nodes, the intermediate hashing results can be precomputed and memoized before the search process, thereby avoiding redundant computations.


\vspace{-4pt}
\stitle{Basic Idea.}
During the intersection test between the \emph{set}-type labeling sets of the query node $v_q$ and the \emph{Bloom filter}-type labeling sets of a candidate node $v$ (lines 11–13 of \algref{alg:dal-bf-query}), the $\texttt{contain}$ operation always maps elements from $v_q$’s labeling sets into hash indices (lines 4–7 of \algref{alg:bf-op}).  
These intermediate hash indices can be memoized prior to the ANN search process to avoid redundant computation.

\begin{algorithm}[t!]

    \SetKwFunction{funcprehash}{preh}
    \SetKwFunction{funcinrangep}{in\_range\_pre}
    \SetKwFunction{funccontainph}{lcontain}
    
    \SetKwProg{Fn}{Function}{:}{end}
    
    \Fn{\funcprehash{$\mathcal{D}(v_q)$}} {
    $P_L,P_t\leftarrow \emptyset$; \\
    \ForEach{$D_i(v_q)\in \mathcal{D}(v_q)$}{
        $P_L^i\leftarrow \emptyset$; \\
        \ForEach{$v\in D_i(v_q)$}{
        $P_L^i[v]\leftarrow \{h_i(v)|h_i\in H\}$; \\
        }
        Append $P_L^i$ to $P_L$; \\
    }
    \Return{$P_L$}; \\
    }

    \Fn{\funcinrangep{$P_L$, $v_o$, $r$}} {
    Lines 1-10 of \algref{alg:dal-bf-query}; \\
    \If{$\texttt{lcontain}({B_j(v), P_L^i})$} {
        \Return{$true$}; \\
    }
    Lines 14-17 of \algref{alg:dal-bf-query}; \\
    }

    \Fn{\funccontainph{$B_j(v)$, $P_L^i$}}{
    \ForEach{$hi\in P_L^i$}{
        $flag\leftarrow true$; \\
        \ForEach{$i\in hi$} {
            \If{$B_j(v)[i]=0$} {
            $flag\leftarrow false$; \\
            \textbf{break} ; \\
            }
        }
        \textbf{if} $flag$ \textbf{then} \Return $true$; \\
        
    }

    \Return $false$;
    }
 \caption{Memoization and Contain Check functions for DLH-M}
 \label{alg:prehash-func}

 \vspace{-2pt}
\end{algorithm}

\algref{alg:batch-nsw-search} presents the pseudocode of \ourindex-M.  
In contrast to the vanilla search process of {\ourindex} (i.e., \algref{alg:nsw-search}), \ourindex-M precomputes the hash representations of query node $v_q$’s labeling sets and stores them as memoized data (line 4).  
Consequently, the in-range query is evaluated using these precomputed hash bits (line 10) instead of original labeling sets of $v_q$.

\algref{alg:prehash-func} presents the procedure for generating the preprocessed hash bits ($\texttt{preh}()$) and for utilizing them ($\texttt{lcontain}()$).  
In $\texttt{preh}()$, the algorithm first computes the hash indices using the hash functions $H$, as shown in line 4 of \algref{alg:batch-nsw-search} and lines 1–8 of \algref{alg:prehash-func}.  
Specifically, each $P_L^i$ stores the hash indices corresponding to the label nodes in $D_i(v_q) \in \mathcal{D}(v_q)$.
When a set intersection between $B_j(v)$ and $D_i(v_q)$ is required, we use the function
$\texttt{in\_range\_pre}()$ to achieve the in-range query.
$\texttt{in\_range\_pre}()$ differs from \algref{alg:dal-bf-query} that it uses
$\texttt{lcontain}(B_j(v),$ $P_L^i)$ for the intersection test between the set-type labels of $v_q$ and Bloom filter-type labels of $v$.
The $\texttt{lcontain}(B_j(v),P_L^i)$ function enumerates the pre-computed hashing indices in $P_L^i$ (lines 16-21) to judge if there exist intersections (line 19) between current labeling sets.

\begin{figure}[t]

\centering

\subfloat[{\ourindex}: Duplicated Hashing Computations.]
{\includegraphics[width=0.41\textwidth]{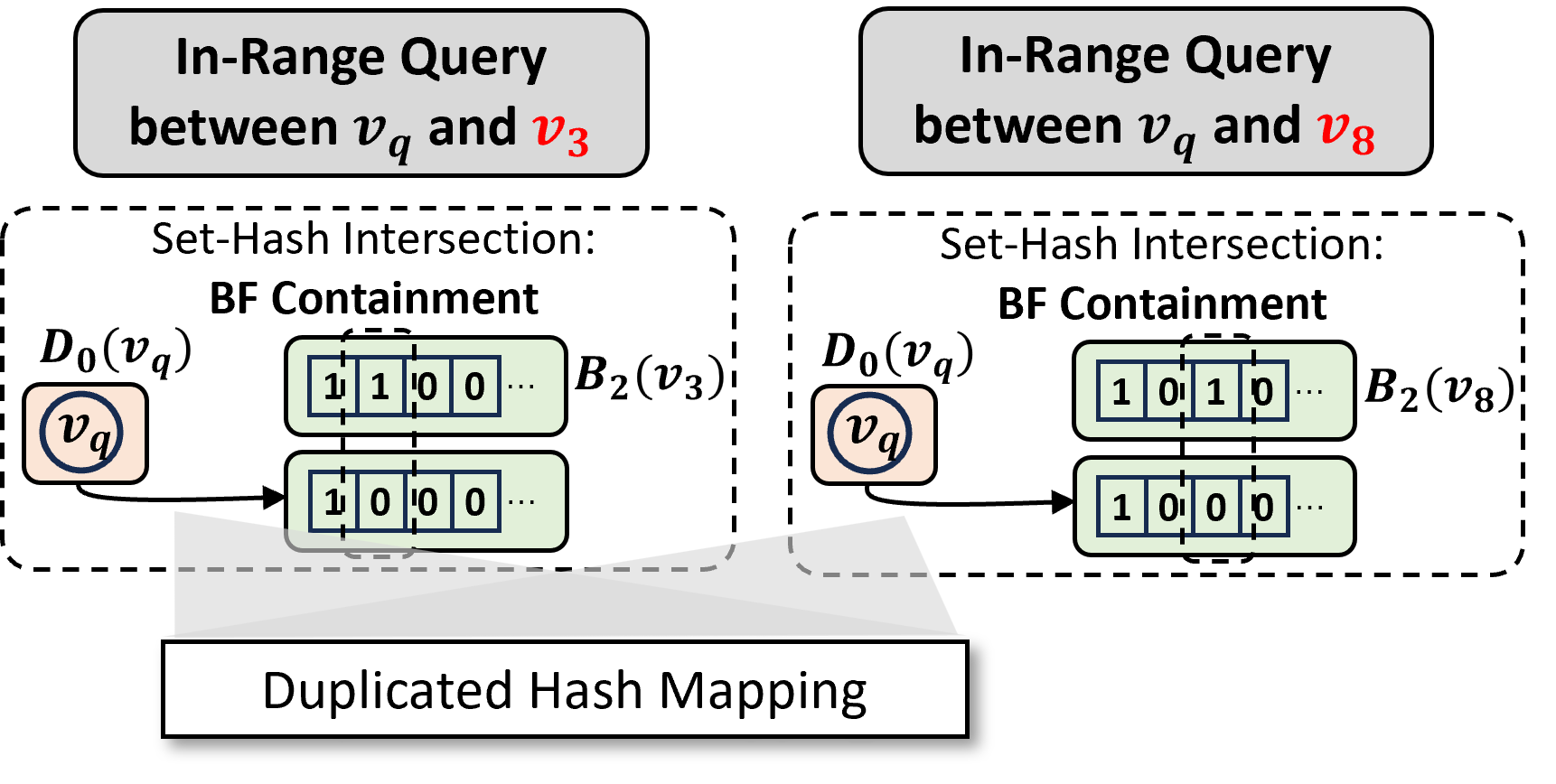}
\label{fig:dlhp-bef}}

\subfloat[{\ourindex}-M: Pre-Hashing as Memoization.]
{\includegraphics[width=0.38\textwidth]{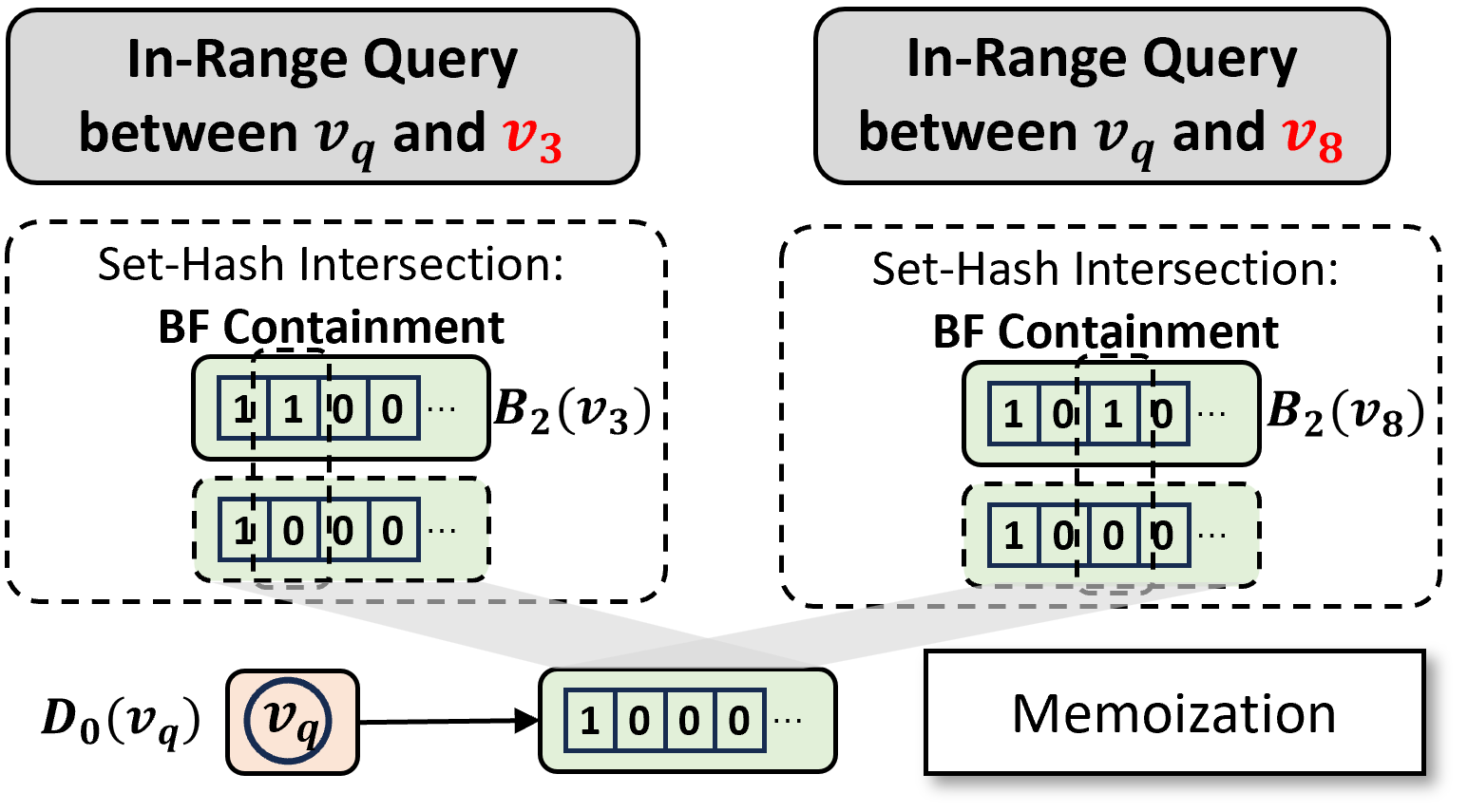}
\label{fig:dlhp-aft}}

\caption{An example of {\ourindex}-M.}
\vspace{-6pt}
\label{fig:dlhp}
\end{figure}

\begin{example}
    Consider the example in \figref{fig:dlhp}.  
Assume that $D_0(v_q)$ of the query node $v_q$ is stored as a set and participates in in-range queries via set–hash intersection with the Bloom filters of different candidate nodes $v_3$ and $v_8$.  
As illustrated in \figref{fig:dlhp-bef}, in \ourindex{} the elements of $D_0(v_q)$ are mapped to Bloom filter indices each time an in-range query is invoked.  
In contrast, as shown in \figref{fig:dlhp-aft}, these hash indices can be precomputed before the query process of {\ourindex}-M, thereby eliminating a large number of redundant hash computations.
\end{example}

\stitle{Time Complexity.}
The memoized intermediate hashing indices can be reused across multiple in-range queries.  
Therefore, we analyze the average time complexity of \ourindex{} over all in-range queries issued by candidate vectors.

\begin{lemma}
    Let $n_c$ denote the number of candidates subject to in-range queries, and let $\overline{t}$ be the average number of retrieved neighbors per query. The average time complexity of an in-range query is then $O(d_{v_q}^<\cdot \overline{t}+d_v^<\cdot t+c_{v_q}^>\cdot c_{v}^>\cdot t)$.
\end{lemma}
\begin{proof}
    The key difference between \ourindex{} and \ourindex{}-M lies in how the set intersection between $D_i(v_q)$ (represented explicitly) and $B_j(v)$ (encoded as a Bloom filter) is performed.
    
    For the $n_c$ in-range queries processed by \ourindex{}-M, memoizing the intermediate hashing indices (i.e., $\texttt{preh}(\cdot)$) takes $O(d_{v_q}^< \cdot t)$ time in total, and each individual in-range query requires $O(d_{v_q}^< \cdot \overline{t})$ time to test containment.  
Consequently, the total time complexity for $n_c$ in-range queries is $O(d_{v_q}^< \cdot t+n_c\cdot d_{v_q}^<\cdot \overline{t})$ for $n_c$ in-range queries.
    Thus, the average time complexity per in-range query is $O(d_{v_q}^<\cdot\frac{t}{n_c}+d_{v_q}^<\cdot \overline{t})$.
    In practice, the number of candidates $n_c$ is often much larger than the fixed number of hash functions $t$, rendering the first term negligible. Hence, the average time complexity simplifies to $O(d_{v_q}^< \cdot \overline{t})$.

    Based on the above discussion and the result in \lemref{lem:time-comp}, the average time complexity of a single in-range query is $O(d_{v_q}^<\cdot \overline{t}+d_v^<\cdot t+c_{v_q}^>\cdot c_{v}^>\cdot t)$.
\end{proof}
\vspace{-4pt}
\stitle{Discussion.}
Notably, {\ourindex}-M guarantees the exact output to {\ourindex}, as the memoization accurately records the hashing indices of the labeling sets of the query node.
These hash indices are then utilized to test whether there exist set intersections between the labeling sets of the query node in \emph{set} type and that of the candidate in \emph{Bloom filter} type, which behaves identically to \ourindex.

\begin{table}[t!]
    \centering
    \caption{Datasets.}
\begin{tabular}
    {cccc}
    \toprule[1pt]

 \tableincell{Dataset} & \tableincell{Category} & \tableincell{\#dim} & Similarity \\
    \midrule[0.8pt]
    SIFT & Image & 128 & Euclidean \\ \cmidrule{1-4}
     GIST & Image & 960 & Euclidean \\ \cmidrule{1-4}
    DEEP & Image & 96 & Euclidean \\  \cmidrule{1-4}
    YFCC & Image & 192 & Cosine \\ 
    \bottomrule[1pt]
    \end{tabular}
    \label{tab:dataset}
\end{table}

\section{Experimental Study}
\label{sec:exp}

\subsection{Experimental Setup}
\stitle{Datasets.}
Following previous works~\cite{icde25time,pami11pq}, four vector datasets are considered in our experiment, as illustrated in \tabref{tab:dataset}.
Specifically, we consider the vector datasets including SIFT~\cite{pami11pq}, GIST~\cite{pami11pq}, DEEP~\cite{cvpr16index}, and  YFCC~\cite{acm16yfcc}.
We generate 1 million vectors from each dataset for the experiments.

{%
    \setlength{\abovecaptionskip}{2pt}%
    \setlength{\belowcaptionskip}{0pt}%

\begin{figure*}
    \centering

    \includegraphics[width=0.9\linewidth]{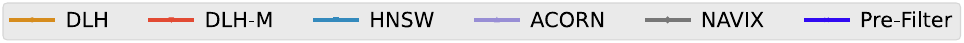}
    \par\vspace{-1mm}

    \subfloat[SIFT ($p=0.00025$)]{%
        \begin{minipage}[b]{0.24\textwidth}
            \centering
            \includegraphics[width=\linewidth]{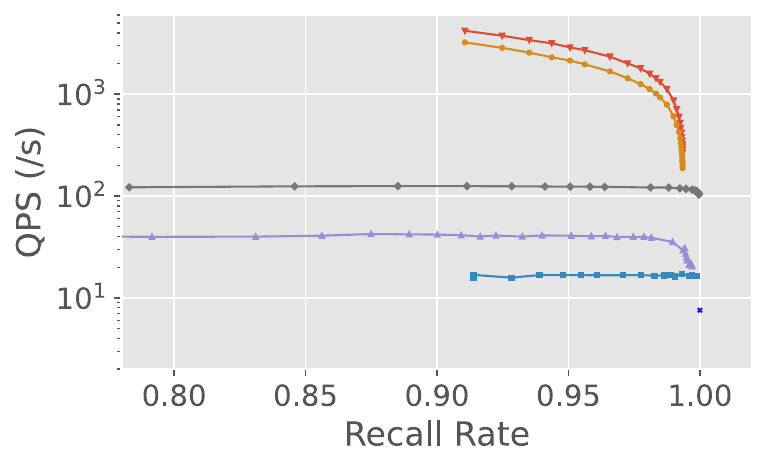}
            \label{subfig:a}
        \end{minipage}
    }\hfill
    \subfloat[GIST ($p=0.00025$)]{%
        \begin{minipage}[b]{0.24\textwidth}
            \centering
            \includegraphics[width=\linewidth]{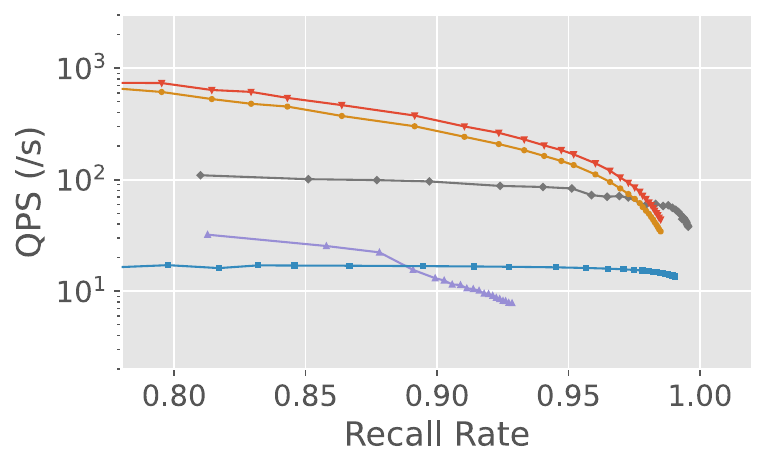}
            \label{subfig:b}
        \end{minipage}
    }\hfill
    \subfloat[DEEP ($p=0.00025$)]{%
        \begin{minipage}[b]{0.24\textwidth}
            \centering
            \includegraphics[width=\linewidth]{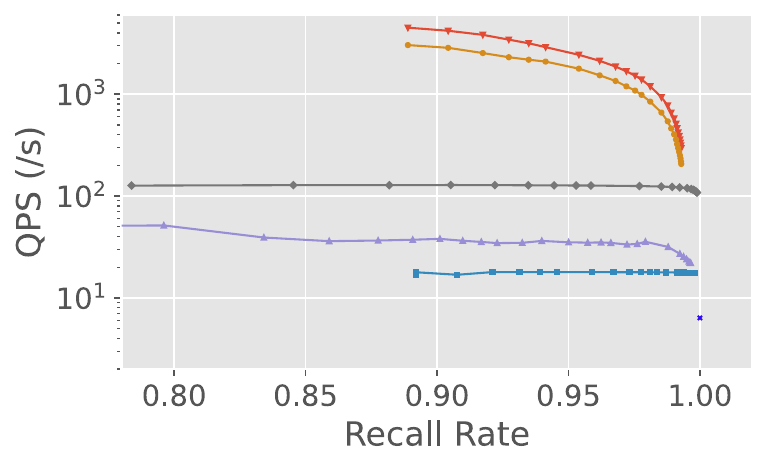}
            \label{subfig:c}
        \end{minipage}
    }\hfill
    \subfloat[YFCC ($p=0.00025$)]{%
        \begin{minipage}[b]{0.24\textwidth}
            \centering
            \includegraphics[width=\linewidth]{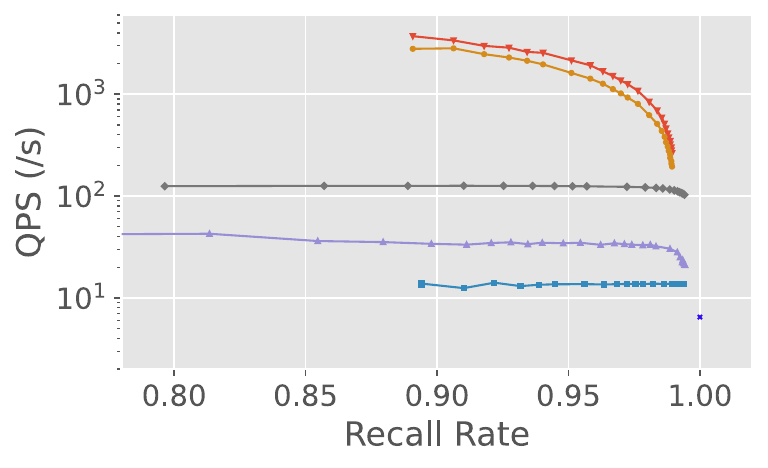}
            \label{subfig:d}
        \end{minipage}
    }

    \par\vspace{-1mm}

    \subfloat[SIFT ($p=0.00027$)]{%
        \begin{minipage}[b]{0.24\textwidth}
            \centering
            \includegraphics[width=\linewidth]{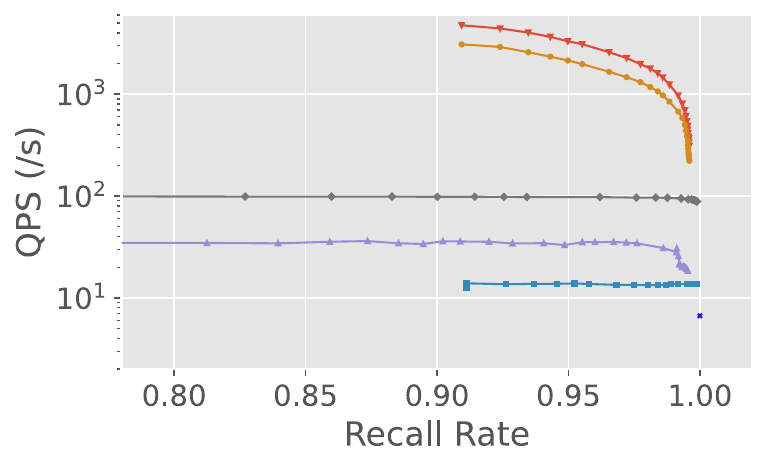}
            \label{subfig:e}
        \end{minipage}
    }\hfill
    \subfloat[GIST ($p=0.00027$)]{%
        \begin{minipage}[b]{0.24\textwidth}
            \centering
            \includegraphics[width=\linewidth]{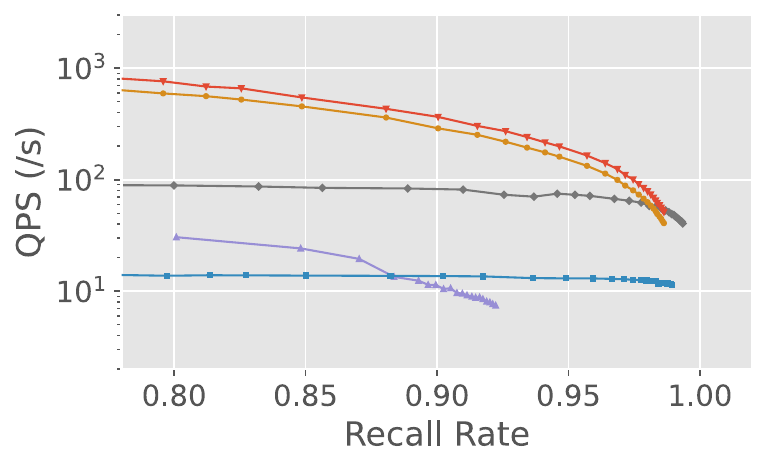}
            \label{subfig:f}
        \end{minipage}
    }\hfill
    \subfloat[DEEP ($p=0.00027$)]{%
        \begin{minipage}[b]{0.24\textwidth}
            \centering
            \includegraphics[width=\linewidth]{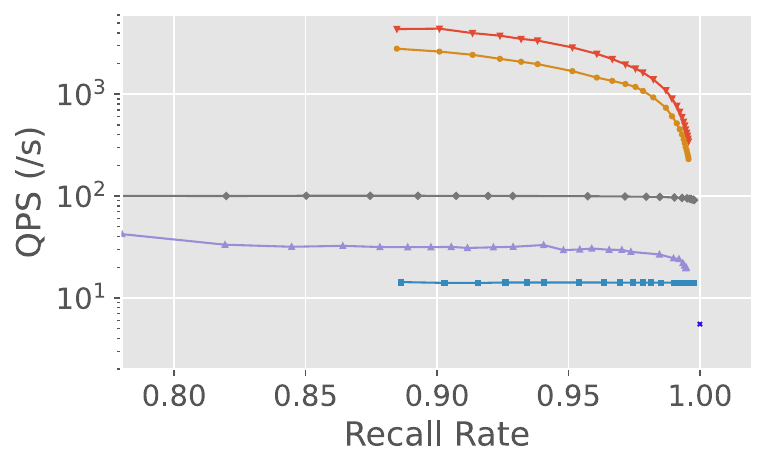}
            \label{subfig:g}
        \end{minipage}
    }\hfill
    \subfloat[YFCC ($p=0.00027$)]{%
        \begin{minipage}[b]{0.24\textwidth}
            \centering
            \includegraphics[width=\linewidth]{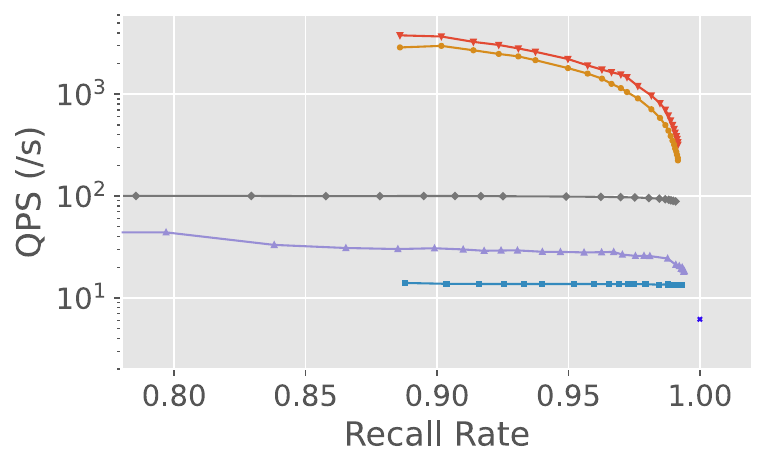}
            \label{subfig:h}
        \end{minipage}
    }

    \par\vspace{-1mm}

    \subfloat[SIFT ($p=0.0003$)]{%
        \begin{minipage}[b]{0.24\textwidth}
            \centering
            \includegraphics[width=\linewidth]{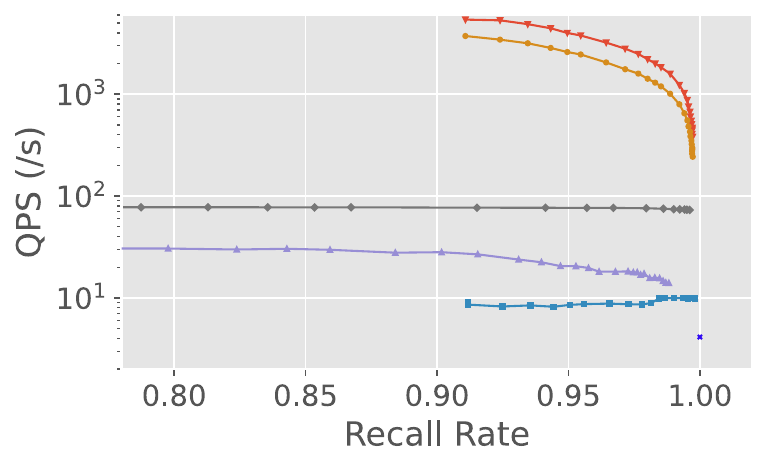}
            \label{subfig:i}
        \end{minipage}
    }\hfill
    \subfloat[GIST ($p=0.0003$)]{%
        \begin{minipage}[b]{0.24\textwidth}
            \centering
            \includegraphics[width=\linewidth]{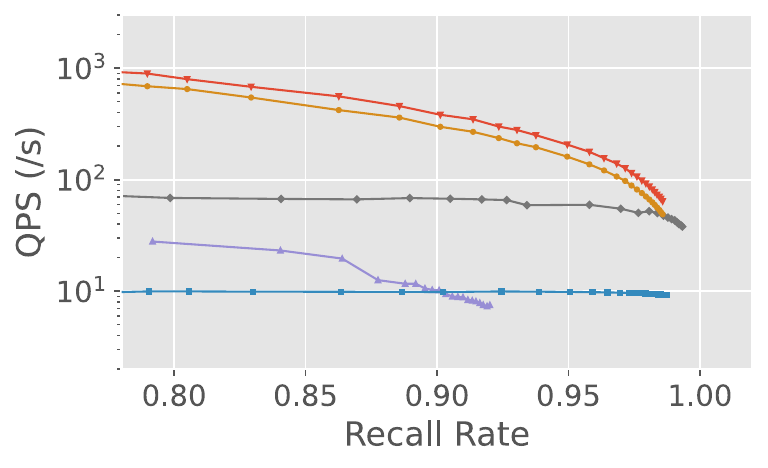}
            \label{subfig:j}
        \end{minipage}
    }\hfill
    \subfloat[DEEP ($p=0.0003$)]{%
        \begin{minipage}[b]{0.24\textwidth}
            \centering
            \includegraphics[width=\linewidth]{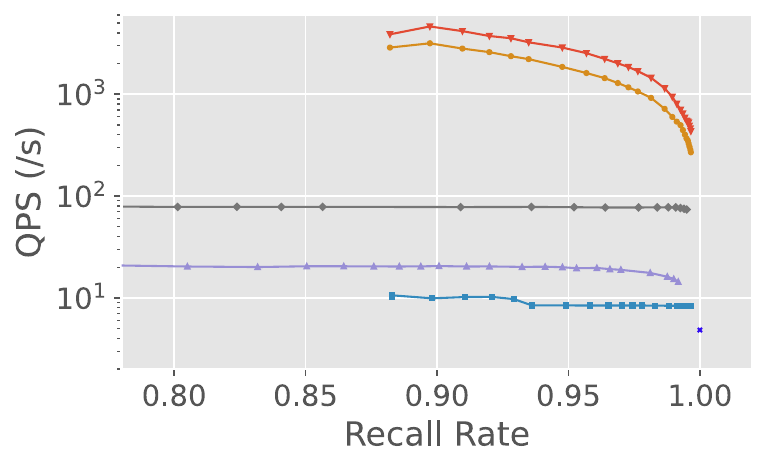}
            \label{subfig:k}
        \end{minipage}
    }\hfill
    \subfloat[YFCC ($p=0.0003$)]{%
        \begin{minipage}[b]{0.24\textwidth}
            \centering
            \includegraphics[width=\linewidth]{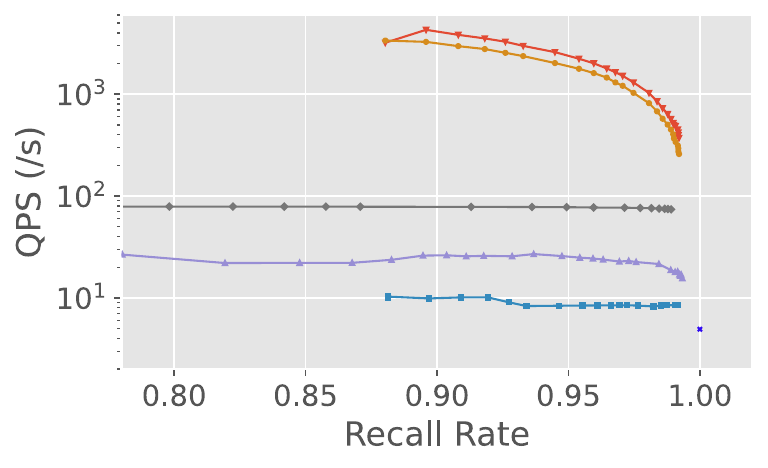}
            \label{subfig:l}
        \end{minipage}
    }
    \vspace{4pt}
\caption{Overall Performance.} \vspace{-4pt}
    \label{fig:exp-search}
  
\end{figure*}
}%

\stitle{Workloads.}
To address the performance of {\ourindex} on the ANNGR problem, we generate a filter graph of $n=80,000$ nodes by the Erdos-Renyi model $G(n,p)$~\cite{erdos59a} for each vector dataset with $p=\{0.00025,$ $0.00027,0.0003\}$.
The nodes are then associated with the vectors uniformly.
For each query, we retrieve the $k$-nearest neighbors ($k=128$) within a graph range of $r \in \{3, 4, 5, 6\}$, where $r=4$ serves as the default parameter.

\stitle{Baselines.}
We compare {\ourindex} and {\ourindex}-M with the following baselines. 
In these baselines, a BFS search on the filter graph is first utilized before the search process to find the $r$-hop neighborhood of the query node.
The obtained node set is then employed as a predicate function for distinguishing whether a candidate locates in $r$-hop neighbors of the query node.
\begin{itemize}
    \item Pre-Filter~\cite{vldb20dbv}: It first obtains those vectors that satisfy the graph range, and then enumerates all result vectors to find the closest vectors. No index is used in this baseline.
    \item Post-Filter (HNSW) ~\cite{pami20hnsw}: It first searches the vectors in the increasing order of distance to the query vector based HNSW, and then validates whether the vectors satisfy the graph filter.
    \item ACORN~\cite{sigmod24acorn}:
    The index supports ANN search with arbitrary filters by a dedicated graph index.
    We choose default configurations $M=32$ and $M_{\beta}=64$, and utilize the inverse of the ratio of nodes that satisfy the filter in the graph.
    \item Navix~\cite{vldb25navix}:
    The index supports ANN with arbitrary filters by integration of HNSW and a prefiltering strategy. 
\end{itemize}


\stitle{Implementation.}
All methods are implemented in C++ with -Ofast optimization. 
For {\ourindex}, the default maximum number of neighbors and beam search width in the ANN index are chosen as $M=16$ and $b=200$, respectively.
For the parameters of the Bloom filters in {\ourindex}, we choose the FPP $\alpha=0.01$.
All experiments are repeated in $5$ times, and the average is reported.

\stitle{Environment.}
All experiments are conducted on a server with Intel Xeon(R) Gold 6240 2.60GHZ CPU processors.

\vspace{-4pt}

\subsection{Experiments on Search Performance}
\label{subsetc:exp-search}

To evaluate the search performance of the proposed algorithms, we compare {\ourindex} and {\ourindex}-P with the baselines considering the following metrics.
\vspace{-2pt}
\begin{itemize}
    \item Queries per second (QPS) defines the number of processed queries per second, measuring the efficiency for ANN search.
    \item Recall rates define the proportion of vectors that satisfy the filter in all of the obtained results.
    It reflects the quality of the ANN search results.
\end{itemize}

We report the results of QPS vs. recall rates by varying beam search width $b$ in \algref{alg:nsw-search} to evaluate the search performance, following previous works~\cite{icde25time,vldb21graphknn,tkde20annexp}.

\stitle{Overall Performance.}
\figref{fig:exp-search} presents the search performance of {\ourindex} and {\ourindex}-M. 
Generally speaking, {\ourindex} and {\ourindex}-M obtain up to 70.3$\times$ QPS under identical recall rate compared to the baselines in most cases, with {\ourindex}-M obtaining higher QPS than {\ourindex}.
Meanwhile, both proposed indices achieve recall rates exceeding $98\%$ across four datasets while maintaining higher QPS than the baselines, demonstrating the strong retrieval performance of proposed indices under the Bloom filter approximation.
In comparison, all of the baselines reveal a decreased and unchanged QPS under various recall rates.
This is attributed to the large part of the time it takes to find the neighbors of $h$-hop of the naive BFS search in the filter graph, which dominates the whole ANN search process.

\begin{figure}
\centering

\subfloat
{\includegraphics[width=0.4\textwidth]{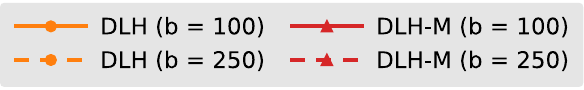}
\label{fig:scala-legend}}
\setcounter{subfigure}{0}
\subfloat[][QPS.]
{\includegraphics[width=0.23\textwidth]{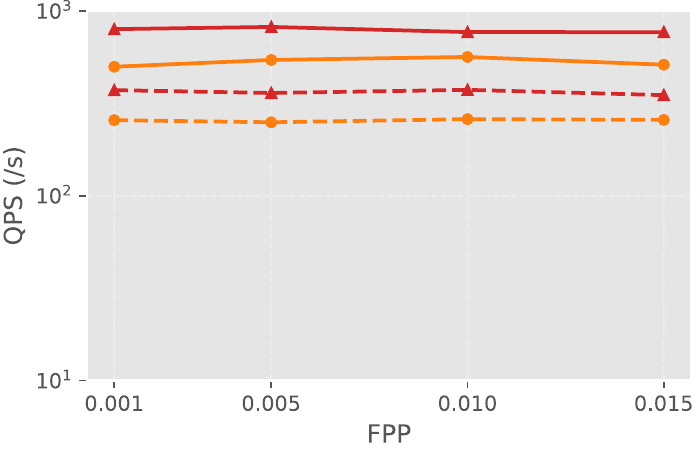}
\label{fig:param-fpp-qps}}
\subfloat[][Recall Rate.]
{\includegraphics[width=0.23\textwidth]{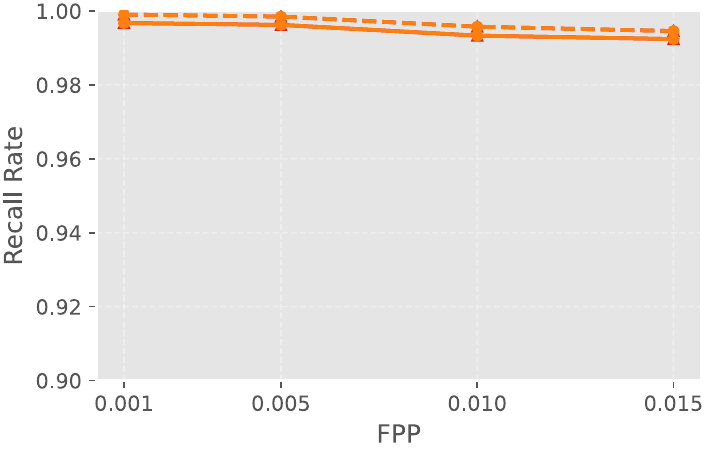}
\label{fig:param-fpp-recall}}

\caption{Experimental Results on Varying FPP.}
\vspace{-10pt}
\label{fig:param-fpp}
\end{figure}

\begin{figure}
\centering

\subfloat
{\includegraphics[width=0.4\textwidth]{figures/scala/scalability_legend.pdf}
\label{fig:scala-legend}}
\setcounter{subfigure}{0}
\subfloat[][QPS.]
{\includegraphics[width=0.23\textwidth]{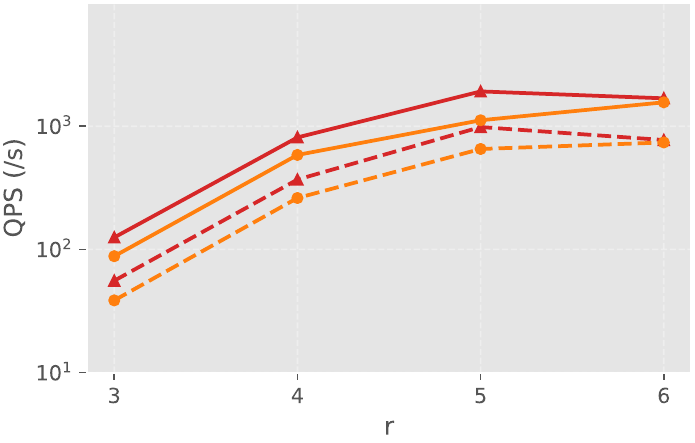}
\label{fig:param-khop-qps}}
\subfloat[][Recall Rate.]
{\includegraphics[width=0.23\textwidth]{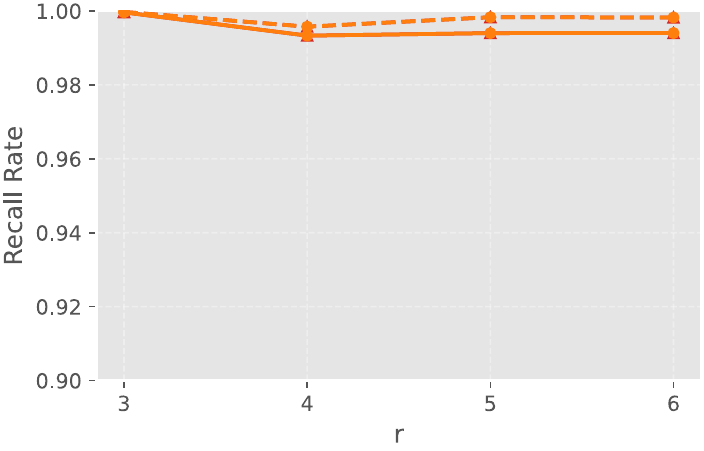}
\label{fig:param-khop-recall}}
\vspace{-4pt}
\caption{Experimental Results on Varying Graph Range $r$.}
\vspace{-8pt}
\label{fig:param-khop}
\end{figure}

\stitle{Performance on Varying Datasets.}
Comparing the performance of the algorithms on different datasets (horizontal comparison of \figref{fig:exp-search}), we observe that {\ourindex} and {\ourindex}-M on SIFT, DEEP and YFCC show a much steeper improvement in the QPS-recall trade-off ({\eg} QPS$>10^3$ with recall $>96\%$).
In contrast, the results on GIST decrease more steadily.
This is attributed to the higher dimensionality of vectors in the GIST dataset, which intensifies the computational overhead of vector similarity calculations. Regardless of the choice of $b$, these computations consistently dominate larger portion of time while the search process. Consequently, this leads to a modest yet steady decline in QPS for the GIST dataset.


\stitle{Performance on Varying Sparsity of Filter Graphs.}
When varying the sparsity of the filter graph via $p$ (vertical comparison of \figref{fig:exp-search}), we observe a slight QPS decrease for the baselines.
This is because, with a larger $p$ ({\ie} larger average degrees in the filter graph), all of the baselines spend more time finding all vectors that satisfy the graph range filter.
In contrast, although larger average degrees of the filter graphs may lead to larger labeling sets in {\ourindex} and {\ourindex}-M, the Bloom filters utilize a similar number of bits for compressed representations of these sets, thus
guaranteeing similar trends for the proposed algorithms.
This validates the well extensibility of {\ourindex} and {\ourindex}-M for more dense filter graphs.

\stitle{Performance on Varying FPP of Bloom Filters.}
\figref{fig:param-fpp} reports the QPS and recall rates of {\ourindex} and {\ourindex}-M on SIFT.
With FPP of Bloom filters varying from $0.001$ to $0.015$, the QPS of both {\ourindex} and {\ourindex}-M stay stable, while the recall rates slightly decrease (still over $99\%$).
This is because a larger FPP can lead to false positives of in-range queries, thus resulting in slight decrease in the recall rate.
The decrease is acceptable, indicating the effectiveness of the employed Bloom filters for hash compression.

\stitle{Performance on Varying Graph Range $r$.}
\figref{fig:param-khop} presents the experimental results on varying the graph range $r$.
With $h$ increases from $3$ to $6$, the recall rates stay stable, validating the advancement of {\ourindex} and {\ourindex}-M.
The QPS continues to grow as $r$ approaches $5$, but stabilizes beyond this point. 
This trend occurs because increasing $r$ improves the proportion of candidate vectors that satisfy the graph range constraint, thereby enhancing the QPS. 
However, when $r$ reaches $5$, the vectors meeting the constraint encompass nearly all nodes within the filter graph. 
Consequently, the QPS for $r=5$ and $r=6$ converges to a similar level. 
This observation holds for both {\ourindex} and {\ourindex}-M, which achieve comparable performance.

\subsection{Experiments on Scalability Test}

\begin{figure}
\centering

\subfloat
{\includegraphics[width=0.4\textwidth]{figures/scala/scalability_legend.pdf}
\label{fig:scala-legend}}
\setcounter{subfigure}{0}
\subfloat[][QPS.]
{\includegraphics[width=0.23\textwidth]{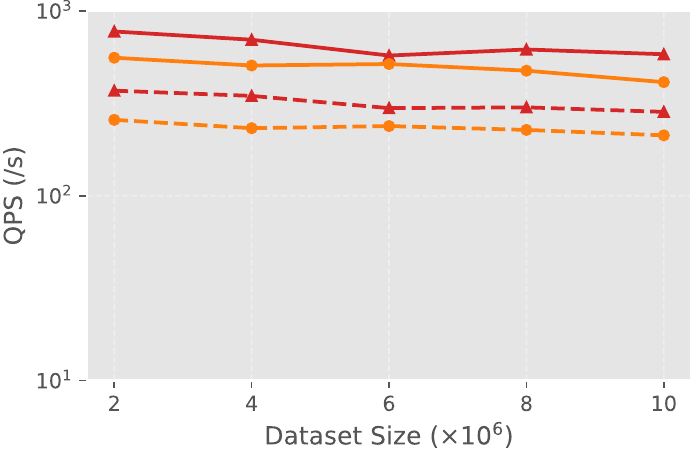}
\label{fig:scala-qps}}
\subfloat[][Recall Rate.]
{\includegraphics[width=0.23\textwidth]{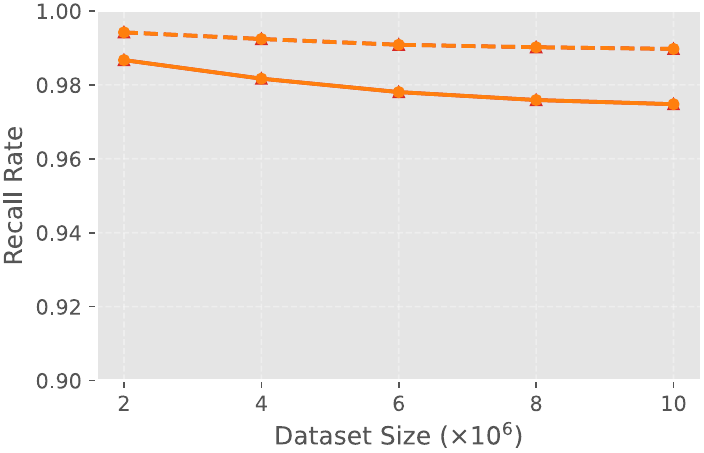}
\label{fig:scala-recall}}

\caption{Search performance of scalability test.}\label{fig:scalability}
\vspace{-8pt}
\end{figure}

\begin{figure*}
\centering

\subfloat{%
\includegraphics[height=0.8cm]{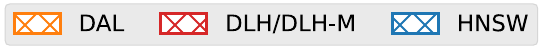}%
\label{fig:index-size-legend}%
}
\vspace{-4pt}
\setcounter{subfigure}{0}

\subfloat[SIFT($p=0.00025$)]{%
\includegraphics[width=0.3\textwidth]{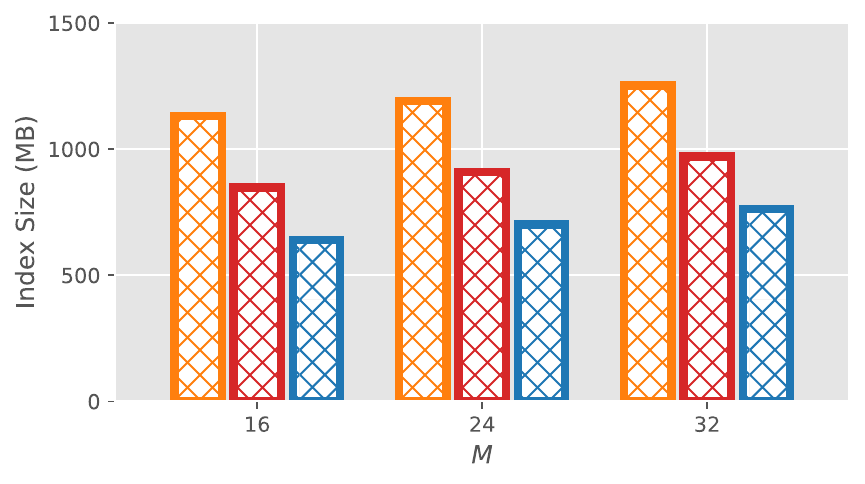}%
\label{fig:index-size-a}%
}\hfill
\subfloat[SIFT($p=0.00027$)]{%
\includegraphics[width=0.3\textwidth]{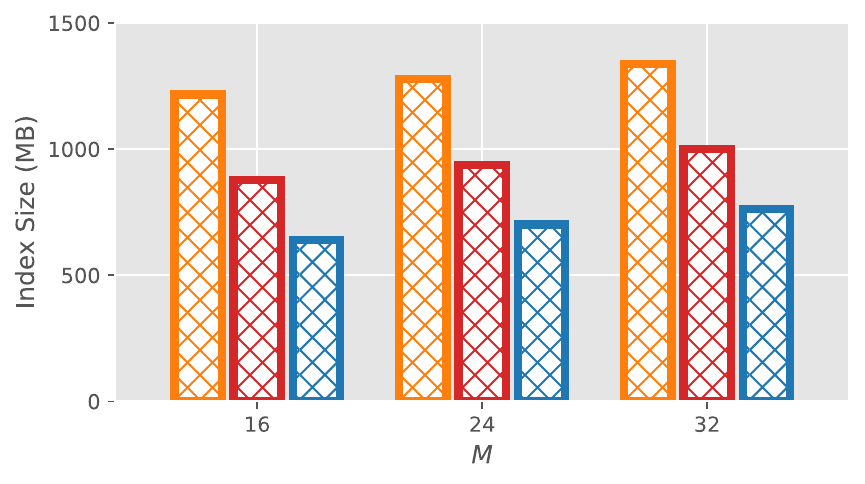}%
\label{fig:index-size-b}%
}\hfill
\subfloat[SIFT($p=0.00032$)]{%
\includegraphics[width=0.3\textwidth]{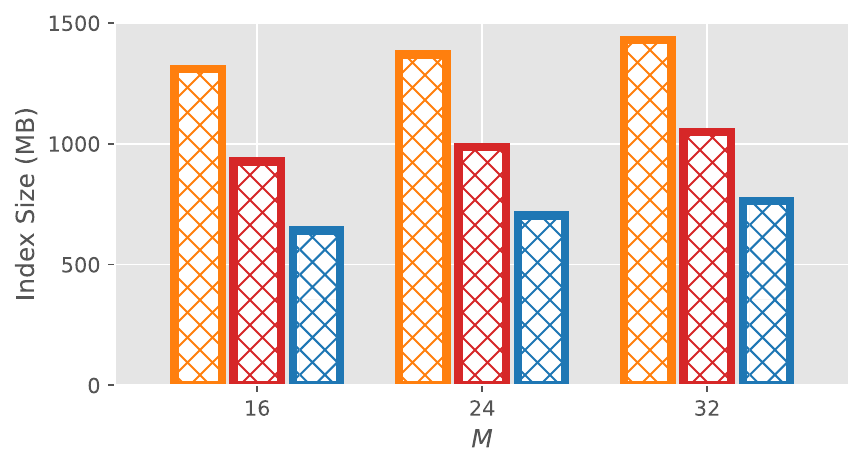}%
\label{fig:index-size-c}%
}
\caption{Indexing construction space overhead on the SIFT dataset.}
\vspace{-4pt}
\label{fig:index-size}
\end{figure*}
 As shown in \figref{fig:scalability}, to test the scalability of the proposed indices, we increase the total number of vectors from $2\times 10^6$ to $10\times 10^6$, and report the QPS and recall rate of {\ourindex} and {\ourindex}-M under $b=100$ and $b=250$.
As shown in \figref{fig:scala-qps}, with the dataset size increasing, the QPS of the proposed indices keep stable.
This is attributed to the design of {\ourindex} that keeps a stable time cost for in-range queries and stable error rate of Bloom filters.
Meanwhile, both {\ourindex} and {\ourindex}-M achieve slightly lower QPS when $b$ increases, as more candidates are considered to test whether they satisfy the graph filter.
Besides, {\ourindex}-M 
keeps identical recall rates to {\ourindex}, because {\ourindex}-M ensures exact output to {\ourindex}.

\subsection{Experiments on Index Construction}
To evaluate the storage performance of the proposed indices, we report the \textbf{index size} of DAL+HNSW (in \secref{subsec:dal}) and {\ourindex}+HNSW (in \secref{sec:alg}) against the original ANN index, {\ie} HNSW, under different sparsity $p$ for the filter graphs.
We omit the results for {\ourindex}-M as it has identical results to {\ourindex}.

\stitle{Overall Performance.}
Figure~\ref{fig:index-size} illustrates the index sizes of the three evaluated indices. 
DAL incurs the highest storage overhead, requiring $1.41\times$ to $1.76\times$ more space than {\ourindex} and HNSW, respectively. 
This significant footprint stems from the inherently complex structure of filter graphs, which necessitate the construction of extensive labeling sets, particularly for large-scale datasets. 
In contrast, {\ourindex} consumes only $1.29\times$ the storage of the original HNSW index. 
This modest increase is primarily due to the integration of Bloom filters for set compression, confirming that {\ourindex} achieves superior space efficiency while maintaining high performance.


\stitle{Performance on Varying $p$.}
With $p$ increasing from $0.00025$ to $0.0003$ ({\ie} the reported results from \figref{fig:index-size-a} to \figref{fig:index-size-c}), the \emph{extra index size} of DAL  and {\ourindex} increases from 500MB and 218MB to 678MB and 296MB, respectively.
The size of {\ourindex} keeps slight compared to the HNSW index size of 630MB.
This validates the flexibility of {\ourindex} under filter graphs with different sparsity.

\stitle{Performance on Varying $M$.}
When varying the maximum degree $M$ in the HNSW ({\ie} various $M$ in \figref{fig:index-size-a}, \figref{fig:index-size-b} and \figref{fig:index-size-c}), the index size of HNSW ({\ie} the blue bars) increases from 630MB to 752MB (the blue bar in \figref{fig:index-size}).
Meanwhile, since the index size of {DAL} and {\ourindex} is independent of the original index, the index size keeps stable under varied $M$.

\subsection{Summary}
We draw the following observations from the evaluation.
\begin{itemize}
    \item Both proposed indices achieve recall rates exceeding 98.5\% while delivering significantly higher QPS than the baselines, demonstrating their strong performance.
    \item The proposed indices exhibit stable and robust performance across varying parameters, including the FPP of the Bloom filter and the graph range in the filter graph.
    \item The indices incur less than 1.29$\times$ the storage overhead of the original vector indices, indicating good space efficiency.
\end{itemize}


\vspace{-2pt}
\section{Related Works}
\label{sec:rel}

\stitle{Attributed Vector Database.}
Similarity search on vectors has been widely adopted in both academics and industry~\cite {corr24faiss,sigmod21milvus,vldb20dbv,qdrant2025url,icde25blend}.
The high dimensionality of the vectors in the database causes challenges in the similarity search process, leading to a research direction of approximate vector similarity search.
Recent efforts have been proposed to achieve efficient approximate similarity search from various perspectives, including hashing compression~\cite{vldb07lsh,tkde24dblsh}, clustering-based pruning~\cite{nips21spann}, quantization~\cite{pami11pq,pami14opq,eccv18lsq,sigmod24rabit}.
Among them, graph-based indices emerge as the preference due to their significant performance~\cite{pami20hnsw,tods25knnrevisit,vldb19navigate,vldb21graphknn,vldb25wolve,vldb25revisit}.

More recent attentions have focused on the attributed vector databases~\cite{corr25filterexp,corr25survey}, {\ie} vectors are associated with certain attributes, and the retrieved similar vectors should satisfy given filters.
Common attributes and filters include numerical attributes and the range filters~\cite{vldb24unify,sigmod24irange,sigmod24serf}, and category attributes and hit filters~\cite{icde25tag,www23diskann,sigmod25rwalks}.
Other studies support arbitrary attributes and filters~\cite{vldb25navix,sigmod24acorn}.

Among these works, those designed for numerical and category attributes cannot be applied to the ANNGR problem in this paper due to the complicated graph range distributions for different query vectors.
Solutions to arbitrary attributes fit for the ANNGR problem, but are limited to the performance bottleneck of finding the $h$-hop range neighborhoods.

\stitle{Index Design for Graph RAGs.}
Recent studies have focused on RAGs with graph-structured relationships~\cite{corr24graphrag,corr25mem0}.
This demand motivates a related research line of index design for graph RAGs.

For example, \cite{nips24gretriever} and \cite{corr24graphrag} map the adjacent information and the k-hop in the graph into vectors, respectively. 
The vectors are then indexed and retrieved as the regular vectors in the vector database.
Other works~\cite{acl23mvp,acl23triplet,corr24unioqa} seek to convert the graph structures into natural language, such that the retrieved context from RAGs contains enriched information from graphs.

These works aim at enriching the graph structure information of RAGs in format of vectors or texts.
In comparison, our paper focuses on the index design for a more specific and practical $k$-NN search problem from the perspective of vector database, which is orthogonal to these works.





\vspace{-4pt}

\section{Conclusion}
\label{sec:con}

This paper investigates the filtered ANN problem with graph range filters, which aims to retrieve the top-$k$ vectors that are both highly similar to the query and satisfy structural constraints in the filter graph. 
We propose {\ourindex}, a post-filtering framework that constructs a lightweight, distance-aware labeling index with Bloom filter-based hashing compression, enabling efficient query processing with low overhead. Observing the asymmetry in in-range queries, we further design an improved index, namely {\ourindex}-M, which exploits memoization of intermediate hashing results to reduce redundant computations and improve query efficiency. 
Extensive experiments on diverse real-world datasets demonstrate that our approaches outperform state-of-the-art baselines, achieving up to a 70.4$\times$ improvement in QPS while maintaining high recall (above $98.5\%$) with modest space overhead making our approach a solution for practical systems that require both high-performance similarity search and flexible filtering capabilities.




\bibliographystyle{ACM-Reference-Format}
\bibliography{sample}

\end{document}